\documentclass[11pt, one column]{article}
\usepackage[table,xcdraw]{xcolor}
\usepackage[margin=1in]{geometry}
\usepackage[subpreambles=true,mode=buildnew]{standalone}
\usepackage{import}
\usepackage{blindtext}
\usepackage[utf8]{inputenc}
\usepackage{amsmath}
\usepackage{amsfonts}
\usepackage{amssymb}
\usepackage{mathrsfs}
\usepackage{amsthm}
\usepackage[english]{babel}
\usepackage{tikz}
\usepackage{mathtools}
\usepackage{algpseudocode}
\usepackage{algorithm}
\usepackage{float}
\usepackage{blindtext}
\usepackage{hyperref}
\usepackage{xcolor}
\hypersetup{
    colorlinks,
    linkcolor={red!50!black},
    citecolor={blue!50!black},
    urlcolor={blue!80!black}
}
\usepackage{mathtools}
\usepackage{todonotes}
\usepackage{subcaption}
\usepackage{caption}
\usepackage{csquotes}
\usepackage{authblk}
\usepackage{graphicx}
\usepackage[square,sort&compress]{natbib}

\allowdisplaybreaks[1]

\newcommand{\cost}{\operatorname{cost}}
\newcommand{\Pos}{\mathcal{P}}

\newcommand{\deltaitr}{\delta_{itr}}

\newcommand{\lineup}[2]{($#1$,$#2$)-line-up election}
\newcommand{\lineups}[2]{($#1$,$#2$)-line-up elections}

\newcommand{\cand}{\operatorname{cand}}

\newcommand\inv[1]{#1\raisebox{1.15ex}{$\scriptscriptstyle-\!1$}}

\newcommand{\C}{\mathcal{C}}
\newcommand{\V}{\mathcal{V}}
\newcommand{\R}{\mathbb{R}}

\theoremstyle{plain}
\newtheorem{thm}{Theorem}[section]
\newtheorem{lem}[thm]{Lemma}

\newtheorem{cor}[thm]{Corollary}

\theoremstyle{definition}
\newtheorem{defn}{Definition}[section]

\theoremstyle{remark}

\DeclarePairedDelimiter\abs{\lvert}{\rvert}%

\title{Metric Distortion of Line-up Elections: \\ The Right Person for the Right Job}
\author{Christopher Jerrett}
\author{Yue Han}
\author{Elliot Anshelevich}

\affil{Department of Computer Science, Rensselaer Polytechnic Institute}
\affil{\{jerrec, hany4\}@rpi.edu, eanshel@cs.rpi.edu}

\date{\today}

\begin{document}

\maketitle
\begin{abstract}
We provide mechanisms and new metric distortion bounds for line-up elections. In such elections, a set of $n$ voters, $m$ candidates, and $\ell$ positions are all located in a metric space. The goal is to choose a set of candidates and assign them to different positions, so as to minimize the total cost of the voters. The cost of each voter consists of the distances from itself to the chosen candidates (measuring how much the voter likes the chosen candidates, or how similar it is to them), as well as the distances from the candidates to the positions they are assigned to (measuring the fitness of the candidates for their positions). Our mechanisms, however, do not know the exact distances, and instead produce good outcomes while only using a smaller amount of information, resulting in small distortion.

We consider several different types of information: ordinal voter preferences, ordinal position preferences, and knowing the exact locations of candidates and positions, but not those of voters. In each of these cases, we provide constant distortion bounds, thus showing that only a small amount of information is enough to form outcomes close to optimum in line-up elections. 
\end{abstract}
\section{Introduction}
Consider the well-known spacial model of voter preferences \citep{arrow1990advances, carroll2013structure, enelow1984spatial, schofield2007spatial}. In this model, both voters and candidates are located in an arbitrary metric space, with the distance from a voter to a candidate representing how much the voter prefers them: the closer the candidate to the voter, the better. These distances could correspond to ideological differences, or to something more concrete, such as when candidates correspond to placing facilities (e.g., new post offices) and voters want a facility close to them \citep{anshelevich2018approximating}. In many such settings, however, it is too difficult, expensive, or impossible to calculate the exact distances from each voter to each candidate, or the voters may be reluctant to provide such detailed information. On the other hand, obtaining {\em ordinal} knowledge about voter preferences (i.e., voter $i$ prefers $A$ to $B$, and $B$ to $C$) is often much easier. This fact gave rise to a large line of work on {\em metric distortion} (see for example \citet{Abramowitz2019Passion, anagnostides2022metric, anshelevich2017randomized, anshelevich2021ordinal, anshelevich2018approximating, Borodin2019primary, caragiannis2022metric, charikar2022metric, cheng2018distortion, ebadian2022optimized, feldman2016voting, ghodsi2019distortion, goel2017metric, gross2017vote, Kizilkaya2022plurality, pierczynski2019approval, skowron2017social, anshelevich2024approvals, charikar2024breaking}, and see \citet{anshelevich2021distortion} for a recent survey), which is a measure of how well a mechanism knowing only ordinal preferences can perform compared to an omniscient algorithm which knows the true distances between voters and candidates. When attempting to minimize distortion, we are concerned with finding a mechanism that does well in minimizing the social cost (the total distances from the voters to the chosen candidate(s), which come from the underlying metric) while only knowing ordinal information, as compared with the true optimum solution. 

Previously, most work on metric distortion has considered selecting a single candidate \citep{Abramowitz2019Passion, anagnostides2022metric, anshelevich2017randomized, anshelevich2018approximating, anshelevich2021ordinal, charikar2022metric, ebadian2022optimized, ghodsi2019distortion, gkatzelis2020resolving, gross2017vote, Kizilkaya2022plurality, pierczynski2019approval, anshelevich2024approvals}, or selecting a committees of $m$ candidates \citep{caragiannis2022metric}. In this work we instead consider {\em line-up elections} \citep{Boehmer2020line-up}. In such elections, there are multiple positions that need to be filled (e.g., different spots on a sports team, different duties in a club or a committee, etc.), and a shared pool of available candidates to fill them. The goal is to choose which candidate will be assigned to each position, with the constraint that a candidate can only hold a single position. Not all candidates are equally fit for every position, as some may be more qualified for some positions over others. Thus, unlike in the work mentioned above, where we only cared about the total distance from voters to the chosen candidates, we are now concerned both about choosing candidates who are qualified for the position and {\em also} about what candidate best represents the views of the voters. 

These types of considerations also occur in facility location settings. For example, we may have several different building sites and plan to build different types of facilities at some of them (e.g., a post office at one, a food bank at another, etc.). When choosing which sites to select for building, we want to choose ones which are close to the customers. On the other hand, some sites may be more suitable for particular facilities, since they are closer to the places where they would get their supplies (see Figure \ref{fig:intro-example}). In fact, our model is actually more general: positions can also be (not necessarily disjoint) sets in the metric space, instead of single points. If the distance to a set $S$ is defined to be the average distance to the items in the set, then all of our results hold for this case. See Appendix \ref{sec:appendix} for details and for further motivating examples. 

\begin{figure}[ht]
        \centering
        \includegraphics[width=.7\textwidth]{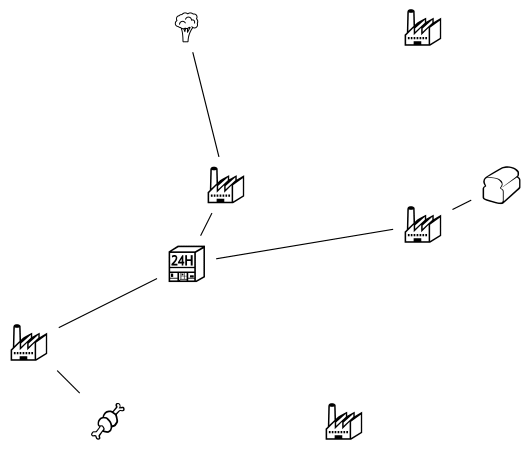}
        \caption{Consider a distribution network where we have farms $\Pos$, warehouses $\C$, and stores $\V$. We show a single store connected to $3$ warehouses, each of which can handle a single type of item. The cost for a single store is represented by the distances to the assigned warehouses and from the warehouses to the farms. Each distribution center can process a single good. We need to form a matching between producers and warehouses so that the total distance all goods need to travel between producers, warehouses, and stores are minimized.}
        \label{fig:intro-example}
\end{figure}

More formally, we have sets of $n$ voters $\V$, $m$ candidates $\C$, and $\ell$ positions $\Pos$ located in some metric space with distance function $d \colon \V \cup \C \cup \Pos \to \R^+$: distances between voters and candidates represent how much the voter likes them, and distances between candidates and positions represent how qualified the candidate is for this position. The goal is to form a matching so that every position has a candidate assigned to it, and every candidate is assigned to at most one position. How much a voter $v\in \V$ likes when a candidate $c\in \C$ is assigned to position $p\in \Pos$ depends on both the similarity between $v$ and $c$, as well as $c$''s fitness for position $p$, and is therefore measured by $d(v,c)+d(c,p)$: the smaller this quantity, the happier the voter $v$ becomes. Thus, the total cost for a voter is the sum of this quantity over all positions and candidates assigned to those positions, and the Social Cost for this model is defined as the sum of all voter costs, as in all previous work mentioned above. 

The optimal solution to this problem is the matching minimizing the social cost, and is simple to compute given the distances or the voter costs. We are interested in forming mechanisms with small distortion, however, and thus assume that we {\em do not know} the distances between the voters and the candidates. Instead we only use a smaller amount of information (see Section \ref{sec:contributions}), and form simple and deterministic voting rules that achieve outcomes only a small factor away from the socially optimal outcome. 

In a sense, line-up elections are a combination of multi-winner elections (as a set of candidates of size $\ell$ is chosen), and bipartite matching (as the chosen candidates are matched with the positions). Because of this, we use techniques optimizing distortion for both multi-winner elections and matching in this paper. As we establish below, however, the fact that both are being optimized simultaneously allows us to form better distortion results than for either in isolation. Many of our proof techniques first show the loss from choosing the incorrect matching, and then combine it with the loss of choosing the incorrect set of candidates. By carefully combining these, however, and showing that the loss of one can often be bounded by the loss of the other, we are able to form new proof techniques and distortion bounds.

Importantly, a limitation of our model is that we require the sets $\V, \Pos$, and $\C$  to all lie in a common metric space. This is a crucial requirement: if we allow for multiple arbitrary metric spaces then it is easy to show that in the worst case our model cannot give bounds better than the standard matching model, which has exponential distortion \citep{caragiannis2022truthful}. That said, there are numerous settings where this assumption holds, such as in distribution networks where all locations $\Pos$ creating the goods, locations $\V$ consuming goods, and intermediate distribution centers $\C$ all lie in a common space. For more discussion and further examples see the full version of this paper.


\subsection{Our Contributions}\label{sec:contributions} 
See Table \ref{tab:results} for a summary of our results.

\renewcommand{\arraystretch}{1.25}
\begin{table}[]
\centering
\begin{tabular}{|l|l|l|l|l|}
\hline
              & $(m, 1)$       & $(m,m)$                & $(m, \ell)$ & (2,1)                        \\ \hline
VP            & 3 (3)           & 3                      & 7 (3)        & 3 (3)                            \\ \hline
PP            & 3 (3)           & $3-\frac{1}{2 ^{m-2}}$ & 5 (3)        & 3 (3)                            \\ \hline
LOC           & 3 (3)           & 1                      & 3 (3)        & 3 (3)                            \\ \hline
VP+PP  & -              & -                      & -           & 2 (2)                         \\ \hline
VP+LOC & $\frac{25}{9}$ & -                      & -           & $\frac{5}{3} (\frac{5}{3})$ \\ \hline
\end{tabular}
\caption{Summary of results showing upper/lower bounds on distortion for lineup elections with different types of available information. VP stands for mechanisms that operate knowing voter preferences, PP for position preferences, and LOC for candidate and position locations. Numbers in parentheses are lower bounds.}
\label{tab:results}
\end{table}


In Section \ref{sec:1} we consider the classic notion of distortion, where only the ordinal voter preferences are known, i.e., we know the order of which candidates each voter prefers for each position. 
It is not difficult to show that for \lineups{m}{1}, i.e., when there is only a single position, our model is equivalent to usual elections where a single candidate is chosen. Thus existing results for distortion in social choice apply for this case. For the special case of \lineups{m}{m} we can also show a distortion of at most 3: this case is simpler than the general one since all candidates must be selected, and we only need to form an appropriate matching without knowing the true distances between the candidates and positions. One of our main results is combining techniques for social choice and matching to give a distortion bound of 7 for general \lineups{m}{\ell}. The mechanism to achieve this bound is simple: it is an iterative election which chooses a candidate to assign to each position separately by running a single-winner election using the voter preferences. Proving the distortion bound, however, requires a careful charging argument between the costs of various components of the matching between candidates and positions, and the distances from voters to candidates. We first show a simpler distortion bound of 9 to build intuition, and then improve it to a bound of 7 by more carefully bounding the combination of loss from choosing the incorrect matching, and choosing the incorrect set of candidates to match to positions.  

We then consider different types of information. In Section \ref{sec:2}, we study mechanisms which know ordinal information about the fit of the candidates for each position (which we refer to as ``position preferences"). 
In social choice, it is often reasonable to know the relative fitness of each candidate for each position (e.g., how much experience they have or their qualifications), even if we don't know the exact fitness levels or the voter preferences. Perhaps surprisingly, this small amount of ordinal information (as compared to knowing the preferences of all the voters) is enough to provide {\em better} distortion than if we knew all voter preferences. In fact, for \lineups{m}{m}, we can form distortion better than the best possible distortion for standard (non-line-up) elections, and for general \lineups{m}{\ell} we can improve our bound from 7 to 5. Once again, the mechanism achieving these bounds is simple: it is serial dictatorship by using the position preferences. We are able to achieve these better results by taking existing bounds on the quality of matching and of social choice, and showing trade-offs between them which combine to form better results than they can separately. 


Finally, in Section \ref{sec:3}, we consider the case where instead we are given the exact locations of all candidates and positions (and thus we know the distances between them), but we know absolutely nothing about the voter preferences. For example, in Figure \ref{fig:intro-example} we may not know the identities or locations of the stores, but have full information about where products are produced and where the distribution centers are located. This information is enough to form a distortion of 3 for \lineups{m}{\ell}, even without knowing anything about voters or their preferences. Interestingly, if we {\em also} know the ordinal preferences of the voters, then for a single position we can improve distortion to $\frac{5}{3}$ for \lineups{2}{1} and to $\frac{25}{9}$ for \lineups{m}{1}. Contrast this with distortion for standard single-winner elections, where knowing the exact locations of all the candidates in addition to ordinal voter preferences does {\em not} improve the possible distortion as compared to only knowing voter preferences \cite{anshelevich2021ordinal}.


\subsection{Related Work}

The line-up election model was first introduced by Boehmer et al. in \citet{Boehmer2020line-up}, where they formalized the concept of a multiwinner election with a shared candidate pool. They analyzed several types of voting rules for such elections from an axiomatic approach and showed several desirable properties of voting rules. 
Our work builds on theirs by considering the distortion of such elections, and forming mechanisms which perform well despite not knowing the exact scores that voters assign to each candidate and position.

Mechanisms for minimizing distortion, i.e., forming mechanisms with access to only ordinal information which perform almost as well as mechanisms with full information, have received a lot of attention. See \citet*{anshelevich2021distortion} for a survey on this topic. Much of this work specifically focuses on {\em metric distortion}, i.e., on forming such mechanisms with the assumption that voters and candidates are located in some unknown metric space (see e.g., 
\citet{anagnostides2022metric, anshelevich2017randomized, anshelevich2018approximating, anshelevich2021ordinal, caragiannis2022metric, charikar2022metric, cheng2017people, feldman2016voting, ghodsi2019distortion, goel2017metric, pierczynski2019approval, skowron2017social, anshelevich2024approvals, kalayci2024proportional}). Most of this work concerns single-winner elections, with \cite{gkatzelis2020resolving, Kizilkaya2022plurality} showing how to achieve a metric distortion of 3 via a deterministic mechanism in single-winner elections. Few papers have looked at distortion of multi-winner elections, with the following main exceptions. 
\citet{chen2020favorite} look at selecting a committee equal in size to the total number of candidates minus one. They show a tight bound of 3 on distortion for the case when the cost of a voter equals the distance to the closest member of the committee. \citet{caragiannis2022metric} also consider the problem of picking a committee of $k$ candidates, but such that the cost for each voter is given by the distance to the $q$'th closest candidate. They then try to minimize the sum of the costs of all agents, and show that the distortion may become arbitrarily large when $q \leq \frac{k}{3}$, asymptotically linear in the number of agents when $\frac{k}{3} < q \leq\frac{k}{2}$ and constant when $q > \frac{k}{2}$. Our model differs from that of \citet{caragiannis2022metric} in a few ways. First, we must both choose $m$ candidates to select, as well as which of these is assigned to which position. Second, the cost of a voter is the sum of its costs for {\em all} positions, so it is the sum of distances to the selected candidates together with their distances to the assigned positions. Thus, our model is both more difficult (because we must form an assignment of candidates to positions as well), and simpler (since we are looking at the sum of distances instead of the $q$'th closest distance). 
Note that if there were no positions and we simply wished to choose $k$ candidates $K \subseteq \C$ to minimize $\sum_{v \in \V} \sum_{c \in K}d(v,c)$, then it is easy to achieve a distortion of 3 by using plurality veto \citep{Kizilkaya2022plurality} to select $k$ candidates one at a time, as in \citet{goel2018relating}.

As discussed above, our model is essentially a combination of classic social choice committee selection, together with forming a matching between selected candidates and positions. The distortion of matching in metric spaces has been previously studied in \citet{anshelevich2016blind, anshelevich2016truthful, bhalgat2011social, caragiannis2022truthful, christodoulou2016social}. Most of this work considers computing the maximum weight matching, however, with the main exception of \citet{caragiannis2022truthful}. They study the problem of minimizing the total social cost of a resource assignment problem, where users are matched with capacitated resources, using 
access only to ordinal preferences. 
As one of their results, they show that the distortion of serial dictatorship is at most $2^k - 1$ for metric bipartite matching. Very recently, \citet{mosesMatching} gave a better algorithm achieving distortion of $O(k^2)$ for this matching problem, as well as a lower bound of $\Omega(\log k)$. In contrast to the above results for matching, by matching candidates to positions at the same time as selecting candidates preferred by voters, and thus considering the total cost of both, in our model we are able to achieve much better distortion.

As discussed above, our model is essentially a combination of classic social choice committee selection, together with forming a matching between selected candidates and positions. The distortion of matching in metric spaces has been previously studied in \citet{anshelevich2016blind, anshelevich2016truthful, anshelevich2019tradeoffs, bhalgat2011social, caragiannis2022truthful, christodoulou2016social}. Most of this work considers computing the maximum weight matching, however, with the main exception of \citet{caragiannis2022truthful}. They study the problem of minimizing the total social cost of a resource assignment problem, where users are matched with capacitated resources, using serial dictatorship with access only to ordinal preferences. They study a resource augmentation framework and show that a $\frac{g}{g-2}$ approximation exists when capacities are multiplied by a factor of $g$. With no augmentation (i.e., given the original resource capacities), they show that the distortion of serial dictatorship is at most $2^k - 1$ for metric bipartite matching. In contrast, by matching candidates to positions at the same time as selecting candidates preferred by voters, and thus considering the total cost of both, in our model we are able to achieve much better distortion.

\section{Model and Preliminaries}


We now formally introduce the \lineup{m}{\ell} and define some terms that will be useful later on in this paper. Let $\Pos$ be a set of $\ell$ positions, $\C$ be a set of $m$ candidates, and $\V$ be the set of $n$ voters in a metric space with distance function $d$. Throughout this paper we will assume $\abs{\Pos} \leq \abs{\C}$. The goal in a line-up election is to form an assignment from the set of positions to the set of candidates, so that each candidate is assigned to at most one position, and each position has a candidate assigned to it. In other words, the goal is to form an injective matching $M \colon \Pos \to \C$. If we assign position $p$ to candidate $c$, then the cost for voter $v$ is given by $d(v,c) + d(c,p)$; this means that the voter cares both about her proximity to the chosen candidate as well as the candidate's fitness for the position. 
Thus the total cost of a voter $v$ for a matching $M$ is defined as the following. 



    \begin{align*}
        \cost(v,d) &= \sum_{p \in \Pos} \left[d(v,M(p)) + d(M(p),p)\right]
    \end{align*}

Thus the total cost of a voter depends both on which candidates are selected, and which positions these candidates are assigned to. The social cost of a matching $M$ is simply the sum of all voter costs, as in most previous work:

\begin{align*}
    \cost(M,d) &= \sum_{v \in \V} \cost(v)
    =\sum_{v \in \V} \sum_{p \in \Pos} \left[d\left(v,M(p)\right) + d\left(M(p), p\right)\right]
\end{align*}
For clarity we will omit $d$ when it is clear from context, and we will denote the set of candidates selected by matching $M$ as $\cand(M) = \{c \in \C \colon \exists p \in \Pos \text{ with } M(p) = c\}$. The goal of a line-up election is to choose an outcome $M$ which minimizes the social cost.

Note that the total cost of $M$ depends on two distinct parts, which we will sometimes consider separately in our proofs. Because of this, we define $\cost_V(M)$ and $\cost_P(M)$ so that $\cost(M)= \cost_V(M) + \cost_P(M)$ and


    \[
        \cost_V(M) = \sum_{v \in \V} \sum_{p \in \Pos} d(v, M(p))
    \]

    \[
        \cost_P(M) = \sum_{v \in \V} \sum_{p \in \Pos} d(M(p),p).
    \]

In this paper we are interested in forming mechanisms which only use a limited amount of information, instead of knowing the exact distances $d$. Knowing the ordinal preferences of the voters for the candidates, instead of their true distances, is the most common type of information used in previous work on distortion. For line-up elections, this means knowing the voter preferences over candidates for each position. In other words, for each position $p$ and voter $v$, and candidates $c_1$ and $c_2$, we know which of these candidates $v$ prefers for position $p$, i.e., which of $d(v, c_1) + d(c_1, p)$ or $d(v, c_2) + d(c_2, p)$ is smaller. We only know the ordering of these when choosing an outcome of the election, however, and not their actual values. 

In later sections, we also consider mechanisms with ordinal knowledge of how qualified different candidates are for each position. In other words, we know, for each position $p$ and candidates $c_1,c_2$, whether $d(c_1,p)$ or $d(c_2,p)$ is smaller, but not the actual distance values. It is often less difficult to sort candidates in order of their qualifications for a position, instead of figuring out the exact numerical level of fitness. We refer to this knowledge as ''position preferences" for the candidates.



Let $I$ be the set of information we have about a line-up election instance. The set $I$ could consist of some combination of ordinal voter preferences, ordinal position preferences, or exact candidate and position locations (for more discussion on how information affects distortion see \citet{anshelevich2024approvals, Abramowitz2019Passion, anshelevich2019tradeoffs, anshelevich2021ordinal}). We now introduce the distortion of a matching $M$ given information set $I$.

\begin{defn}
    Let $\mathcal{D}$ be the set of all metric spaces that are consistent with information set $I$ and $M^*_{d'}$ be the optimal matching for metric $d'$. The {\em distortion} of a matching $M$ is
    \[
    \sup_{d' \in \mathcal{D}} \frac{ \cost(M,d')}{\cost(M^*_{d'},d')}.
    \]
\end{defn}
In other words, distortion of $M$ is the worst case ratio between the cost of $M$ and the optimal solution over any distances which are consistent with the information we possess. Similarly, the distortion of a mechanism is the worst-case distortion of any matching $M$ produced by this mechanism. In this paper, we only focus on the distortion of deterministic mechanisms, and leave the analysis of randomized mechanisms for future work.


\section{Knowing Ordinal Voter Preferences} \label{sec:1}
We first consider the case where for each position, we only know the ordinal preferences of each voter towards each candidate for this position. In other words, for each position $p$ and voter $v$, we know which candidate would be $v$'s first choice for $p$ based on $v$'s cost, which would be the second choice, etc. First we will show some easy reductions to well known problems, as well as some lower bounds, and then show more difficult and general theorems afterward.

We will call the following a {\em standard election}, to differentiate it from line-up elections which are the focus of this paper. Let $\C$ be a set of candidates and $\V$ be a set of voters in some metric space $d$. For any $c \in \C$, let the social cost for each voter $v$ be $d(v, c)$. Then the total social cost of the election is the sum of the individual costs for each voter, $\sum_{v \in \V} d(v,c)$.  The goal is to find a candidate $w \in \C$ such that $w$ minimizes the total social cost.  Having defined the standard election we can now present our first simple result for the case of a single position.

\begin{lem}
    For each \lineup{m}{1} there exists a standard election with the same candidate set $\C$, such that choosing $c \in \C$ in the line-up election is equal to the cost of choosing $c$ in the standard election.
    \label{lem:equal-cost-reduction}
\end{lem}

\begin{proof}
    We shall proceed by reducing the line-up election to a standard election where we instead only try to minimize the total distances from voters to the selected candidate. Suppose we are given a metric space $\mathcal{X}$ with distance function $d$ on $\V \cup \C \cup \Pos$, with $\Pos = \{p\}$. Define a new metric on $\V \cup \C$
    \[
    d'(x,y) = \begin{cases}
                d(v_1,v_2) & \text{if } v_1 \in \V \text{ and } v_2 \in \V\\
                d(c_1,c_2) +  \abs{d(c_1,p) - d(c_2,p)}& \text{if } c_1 \in \C \text{ and } c_2 \in \C\\
                d(v,c) + d(c,p) & \text{if } v \in \V \text{ and } c \in \C
            \end{cases}.
    \]
    The metric $d'$ can be thought of as taking the original space and adding an extra dimension orthogonal to the original metric space where the height of a candidate $c$ is $d(c,p)$. The cost for choosing $c$ is then $\sum_{v \in \V} [d(v,c) + d(c,p)]$ which is identical to the cost of the line-up election. Thus the cost of a matching does not change under this transformation. Next we will show that this distance function is a metric.
    
    We note that this new metric space with $d'$ is a special case of the 1 product metric on two metric spaces. The first metric space is the space we are originally given, $(\mathcal{X},d)$ and the second space is the metric space $\R_{\geq 0}$ where the coordinate in $\R_{\geq 0}$ is zero if the point is a voter and $d(c,p)$ for a candidate $c$. Thus the cost remains unchanged and the new space is still a metric.
\end{proof}

In the previous proof we do not use the fact that the distances from $c$ to $p$ form a metric, so Lemma \ref{lem:equal-cost-reduction} is actually more general. If the function $d$ is a metric only on $\V \cup \C$ then the statement still holds. So we could use arbitrary costs in place of distances from points in $\C$ to $p$. We can use this lemma in conjunction with any of the previously studied distortion mechanisms for standard elections \citep{anshelevich2017randomized,  anshelevich2018approximating, goel2017metric} to get the next lemma.

\begin{lem}\label{lem:equal-dist}
    For any mechanism that achieves a distortion of $\delta$ for the standard election there exists a mechanism that achieves a distortion of $\delta$ for the \lineup{m}{1}.
\end{lem}

\begin{proof}
    By Lemma \ref{lem:equal-cost-reduction}, we know that under any metric space $d$ for the \lineup{m}{1}, there exists a metric space $d'$ such that for each candidate $c \in \C$ the corresponding candidate in $d'$ has the same cost as $c$ in the line-up election.  In addition, the preferences of voters must stay the same since $d'(v,c) = d(v,c) + d(c,p)$ for position $p$ and voter $v$. Then, to solve the original \lineup{m}{1}, we can instead run a standard election with candidates $\C$ and the same voter preferences. If we choose candidate $c$ for the standard election we obtain the same cost for the line-up election by choosing $c$ for the position. Therefore, if we have a mechanism that achieves a distortion of $\delta$ of the standard election, by using the reduction above, for the resulting solution $c$ and the optimal solution $c^*$, we must have $\sum_{v \in \V} \left[d(v,c) + d(c,p) \right]\leq \delta \sum_{v \in \V} \left[d(v,c^*) + d(c^*,p)\right]$.
\end{proof}

From the previous lemma and previous work from \citet{gkatzelis2020resolving, Kizilkaya2022plurality} there exists a mechanism that gives a distortion of three, so we arrive at the next theorem.

\begin{thm}\label{lem:1-pos-k-cand-3}
There exists a mechanism for the the \lineup{m}{1} that achieves a distortion of $3$. Moreover, no other deterministic mechanism can achieve a worst-case distortion bound less than 3, even on a line with 2 candidates.
\end{thm}

    \begin{figure}
        \centering
        \begin{minipage}{.5\textwidth}
          \centering
        \includegraphics[width=.6\textwidth]{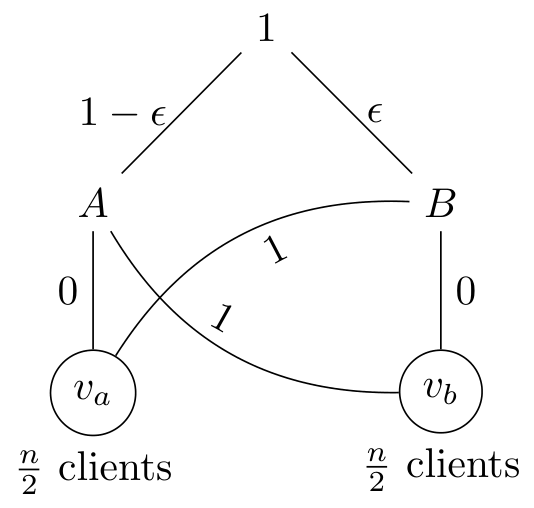}
        \subcaption{Metric space $d_1$.}
        \end{minipage}%
        \begin{minipage}{.5\textwidth}
          \centering
        \includegraphics[width=.6\textwidth]{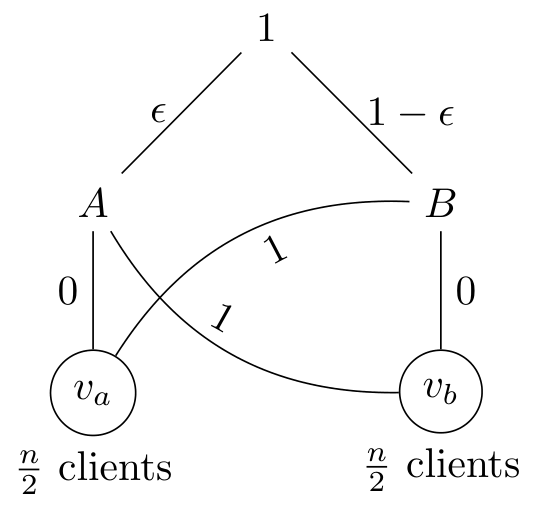}
        \subcaption{Metric space $d_2$.}
        \end{minipage}
        \caption{Two consistent metric spaces with distortion approaching 3.}
        \label{fig:thm11}
    \end{figure}
\begin{proof}
    From \citet{gkatzelis2020resolving, Kizilkaya2022plurality} we know that there exists a deterministic mechanism to solve the standard election which achieves a distortion of $3$. Then, by Lemma \ref{lem:equal-dist}, we can conclude that there exists a mechanism for the \lineup{m}{1} that achieves a distortion of $3$.

     Now we will show a lower bounds of $3$ for all deterministic mechanisms. Consider the following two metrics in $\R$ in Figure \ref{fig:thm11} with $\V = v_a \cup v_b$, $\C = \{A,B\}$ and $\Pos = \{1\}$. Both metric spaces have the same preference profiles for all voters ($v_a$ prefers $A$, and $v_b$ prefers $B$), and any deterministic mechanism must pick the same candidate for both metric spaces. If we pick $A$ then metric $d_1$ gives us a cost $\frac{n}{2}\left(1-\epsilon + 2-\epsilon\right) = \frac{n}{2}\left(3-2\epsilon\right)$ while the optimal candidate is $B$ with cost $\frac{n}{2}(1+2\epsilon)$. So we get a distortion of
    \[
    \lim_{\epsilon \to 0}\frac{\frac{n}{2}(3 - 2\epsilon)}{\frac{n}{2}(1+2\epsilon)} = 3.
    \]

    Similarly, if we pick candidate $B$ we get the same distortion for metric space $d_2$. Thus no mechanism can perform better than $3$.
\end{proof}

Therefore, by Theorem \ref{lem:1-pos-k-cand-3} we know that we can solve \lineup{m}{1} and achieve the optimal distortion of $3$ using standard techniques from metric distortion. However, in Lemma \ref{lem:equal-dist} we also consider distortions other than 3 for general mechanisms since mechanisms such as plurality veto, which achieves optimal distortion, may not have desirable properties such as strategy-proofness. In addition, the mechanism we chose may also be given additional information beyond what we have assumed so far in this section, which would lead to distortions other than three, for example when information about the decisiveness of voters \cite{anshelevich2017randomized} or when information about each voter's passion for each candidate pair is known \cite{Abramowitz2019Passion}. Now that we have discussed results for the case where there is only one position, we will then look at cases where we have more than one position. Before we look at the general case where we have arbitrary $\ell$ positions and $m$ candidates,  we first consider the special case where $\ell = m$.

\begin{lem}\label{lem:all-perfect-matching-3}
If $\abs{\C} = \abs{\Pos}$ then for any two matchings $M_1$ and $M_2$, we have $\cost(M_1) \leq 3\cost(M_2)$.
\end{lem}
\begin{proof}
    Consider the cost of $M_1$:
    \begin{align}
        \cost(M_1) &= \sum_{v \in \V} \sum_{p\in \Pos} \left[d(v, M_1(p)) + d(M_1(p),p)\right]\\
        &\leq \sum_{v \in \V}\sum_{p\in \Pos} \left[d(v, M_1(p)) + d(M_1(p),v) + d(v, M_2(p)) + d(M_2(p),p)\right], \label{kk3label1}
    \end{align}

    where the last line follows from the triangle inequality on $d(M_1(p),p)$. Notice that we have $\cand(M_1) = \cand(M_2)$ so the sum of the distances from voters to candidates in each matching $M_1$ and $M_2$ ($\cost_V(M_1) = \cost_V(M_2)$) are the same in \eqref{kk3label1}, thus we can replace $M_1(p)$ with $M_2(p)$. So we have
    \begin{align}
        \cost(M_1) &\leq \sum_{v \in \V}\sum_{p\in \Pos} \left[d(v, M_1(p)) + d(M_1(p),v) + d(v, M_2(p)) + d(M_2(p),p)\right])\\
        &= \sum_{v \in \V}\sum_{p\in \Pos} \left[d(v, M_2(p)) + d(M_2(p),v) + d(v, M_2(p)) + d(M_2(p),p)\right]\\
        &= \sum_{v \in \V}\sum_{p\in \Pos} \left[3d(v, M_2(p)) + d(M_2(p),p)\right] \\
        &=3\cost_{V}(M_2) + \cost_{P}(M_2) \label{perfect-mtch-3vm}\\
        &\leq 3 \cost(M_2).
    \end{align}
\end{proof}
Notice that $\cost_P(M^*)$ in the transformation above is not multiplied by a factor, while $\cost_V(M^*)$ is multiplied by a constant factor of $3$. By including the voters in our model we can bound the distances between candidates and positions using the distance from the voters to the candidates, allowing us to achieve a constant factor approximation. In contrast, no such mechanism is currently known for the the analogous bipartite perfect matching problem.

We present a concrete example of the above transformation for when $\C = {A,B,C}$, $\Pos = \{1,2,3\}$ and $M = (A,B,C)$ and $M^* = (B, C, A)$ in Figure \ref{fig:example-k-k-3}. Each line represents a factor of 1. We will use diagrams like this one to visualize examples of more complicated proofs later on in this paper.
\begin{cor}
For any \lineup{m}{m} all matchings are at most a factor of 3 away from optimal.
\end{cor}
\begin{figure} [h]
  \begin{minipage}[b]{0.47\linewidth}
    \centering
    \includegraphics{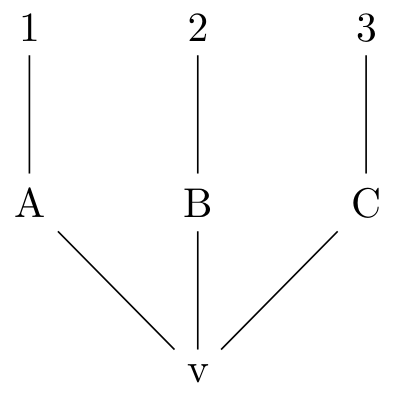}
    \subcaption{The initial matching with the lines representing the cost for each voter $v$ where $\cost(v,M) = d(A,1) + d(B,2) + d(C,3) + d(v, A) + d(v,B) + d(v,C)$}.
    \vspace{4ex}
  \end{minipage}
  \hfill
  \begin{minipage}[b]{0.47\linewidth}
    \centering
    \includegraphics{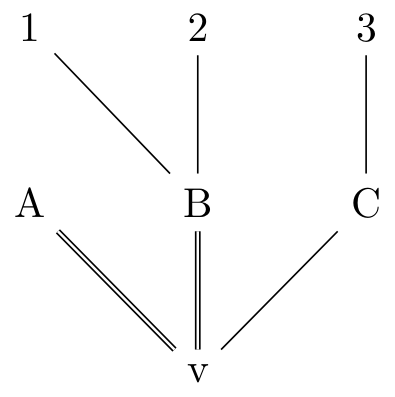}
    \subcaption{In line (\ref{kk3label1}) we use the triangle inequality on $d(A,1) \leq d(A,v) + d(v,B) + d(B,1)$}
    \vspace{4ex}
  \end{minipage} 
  \begin{minipage}[b]{0.47\linewidth}
    \centering
    \includegraphics{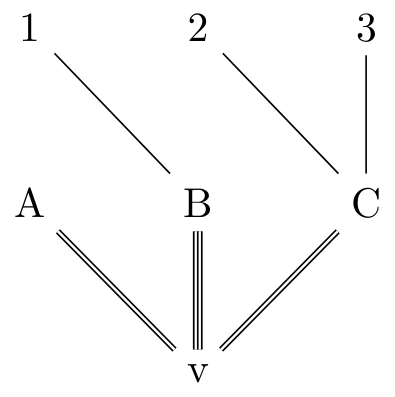}
    \subcaption{We repeat the previous step for each position.}
    \vspace{4ex}
  \end{minipage}
  \hfill
  \begin{minipage}[b]{0.47\linewidth}
    \centering
    \includegraphics{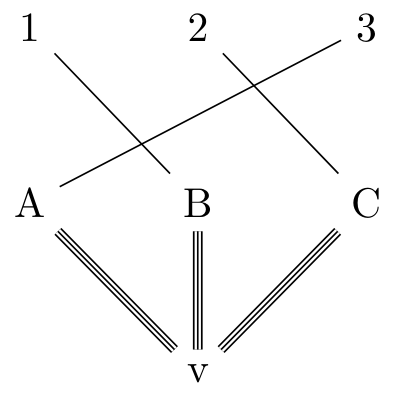}
    \subcaption{We arrive at the desired matching.}
    \vspace{4ex}
  \end{minipage}
  \caption{Transforming $M_1=(A,B,C)$ into $M_2=(B,C,A)$}
  \label{fig:example-k-k-3} 
\end{figure}

After having looked at special cases of $(m,1)$ and \lineups{m}{m}, we will now look at a mechanism for the more general \lineup{m}{\ell} problem. We introduce {\em iterative election}, a natural mechanism for the \lineup{m}{\ell} problem. For a position $p$ reduce the problem to a \lineup{m}{1} problem using the transformation given in Lemma \ref{lem:equal-cost-reduction}, and then use a  mechanism for the standard election problem. If $c$ is selected using the standard election mechanism we assign $M_1(p) = c$. We then remove the selected candidate $c$, and repeat the procedure for each remaining position, removing the selected candidate each time. Importantly, this mechanism is quite general, as we do not assume anything about the mechanism or the specific information set $I$. One could achieve optimal metric distortion for each position by using plurality veto \cite{Kizilkaya2022plurality} for instance, or use a different mechanism that has attractive properties like strategyproofness. Lastly, iterative election is also an online algorithm, meaning that positions can arrive overtime, and as long as we have the appropriate information, we can immediately assign each position at the time of arrival. We will now analyze the case where we utilize a standard election mechanism with distortion $\deltaitr$ and show an upper bound for the iterative election procedure on the \lineup{m}{\ell} problem.

\begin{lem}\label{lem:loser-delta-dist}
    If during the iterative election procedure position $p$ is matched to candidate $c$ utilizing a mechanism with distortion $\deltaitr$, then for any $c'$ that was unassigned at that time we have $\sum_{v \in \V} [d(v,c) + d(c,p)] \leq \deltaitr \sum_{v \in \V} [d(v,c') + d(c',p)]$.
\end{lem}

\begin{proof}
    By Lemma \ref{lem:1-pos-k-cand-3} we have $d'(v,c) = d(v,c) + d(c,p)$, in addition we know that since the mechanism achieves a distortion of $\delta_{itr}$ then $\sum_{v \in \V} d'(v,c) \leq \deltaitr \sum_{v \in \V} d'(v,c')$. Thus $\sum_{v \in \V} \left[d(v,c) + d(c,p) \right] \leq \deltaitr \sum_{v \in \V} [d(v,c') + d(c',p)]$ as desired.
\end{proof}

\begin{lem}\label{lem:three-match-candate-set}
    Let $S \subseteq \C$ be any set of candidates such that $\abs{S} = \abs{\Pos}$, and let $M_1$ be the matching returned by iterative election with mechanism of distortion $\deltaitr$. Then there exists a matching $M_2$ with $\cand(M_2) = S$ such that $\cost(M_1) \leq \deltaitr \cost(M_2)$.
\end{lem}
\begin{proof}
    Let $S_1 = \cand(M_1) \setminus S$, $S_2 = \cand(M_1) \cap S$, and $S_3 = S \setminus \cand(M_1)$. Since $\abs{S_1} = \abs{S_3}$ then there exits a one-to-one function $f:S_1 \to S_3$. We can then conclude that the cost of the solution is:
    \begin{align}
        \cost(M_1) &= \sum_{v \in \V} \sum_{p\in \Pos} \left[d(v, M_1(p)) + d(M_1(p),p)\right]\\
        &= \sum_{v \in \V} \left(\sum_{p\in \Pos , M_1(p) \in S_2} \left[d(v, M_1(p)) + d(M_1(p),p)\right] + \sum_{p\in \Pos , M_1(p) \in S_1} \left[d(v, M_1(p)) + d(M_1(p),p)\right]\right) \label{matching-m-n-3-canidate-set-split1}
    \end{align}
    Now we construct $M_2(p) = \begin{cases}
        M_1(p) & \text{if } M_1(p) \in S_2\\
        f(M_1(p)) & \text{if } M_1(p) \in S_1
    \end{cases}$. Note that this is a valid matching with $\cand(M_2) = S$. We can see that $M_2$ is a valid matching with $\cand(M_2) = S$ since for any $p$ such that $M_1(p) \in S_2$ then $M_2(p) \in S_2$; in addition if $M_1(p) \in S_1$ then $f(M_1(p)) \in S_3$, so we can conclude $\cand(M_2) = S$ and since both $f$ and $M_1$ are one-to-one then so is $M_2$.  By Lemma \ref{lem:loser-delta-dist} we have that $\sum_{v \in \V} d(v, M_1(p)) + d(M_1(p), p) \leq \delta_{itr} \sum_{v \in \V} d(v, f(M_1(p))) + d(f(M_1(p)), p)$ for all $M_1(p) \in S_1$. Thus we can reduce line (\ref{matching-m-n-3-canidate-set-split1}) above to

    \begin{align*}
         \cost(M_1) &\leq  \sum_{v \in \V} \left( \sum_{p\in \Pos , M_2(p) \in S_2} \left[d(v, M_2(p)) + d(M_2(p),p)\right] + \delta_{itr}\sum_{p\in \Pos , M_2(p) \in S_3} \left[d(v, M_2(p)) + d(M_2(p),p)\right] \right)\\
         &\leq \delta_{itr}\cost(M_2)
    \end{align*}
    as desired.
\end{proof}
\begin{thm}\label{thm:3-delta}
    Let $M_1$ be the matching returned by the iterative election procedure and $M^*$ be the optimal matching minimizing the social cost. Then $\cost(M_1) \leq 3\delta_{itr}\cost(M^*)$.
\end{thm}
\begin{proof}
    Let $M_1$ and $M^*$ be defined as above and $M_2$ be the matching given by Lemma \ref{lem:three-match-candate-set} for $S=\cand({M^*})$. Then,
    \begin{align}
        \cost(M_1) &\leq \delta_{itr}\cost(M_2)\\
        &\leq 3\delta_{itr} \cost(M^*).
    \end{align}
    Where the first line follows from Lemma \ref{lem:three-match-candate-set}, and the second line comes from Lemma \ref{lem:all-perfect-matching-3}.
\end{proof}

We know that plurality veto gives the best possible upper bound on the distortion of standard elections; thus if we use it as the base mechanism we can obtain the following corollary to Theorem \ref{thm:3-delta}.

\begin{cor} \label{cor:veto3}
Iterative election returns a matching with distortion at most 9 if we use plurality veto as the base mechanism.
\end{cor}
\begin{proof}
    From \citet{Kizilkaya2022plurality} plurality veto gives at most a distortion of $3$ so we can set $\delta_{itr} = 3$ and get a distortion of $9$.
\end{proof}
However, with a far more careful analysis, we can further improve the distortion factor, which results in the following theorem. 
\begin{thm} \label{thm:2b1-dist-knowing-preferences}
    Let $M_1$ be the matching returned by the iterative election procedure using a mechanism with distortion $\delta_{itr}$, and $M^*$ be the optimum matching. Then $\cost(M_1) \leq (2\delta_{itr} + 1)\cost(M^*)$.
\end{thm}

\begin{proof}
    The proof follows from the more general Theorem \ref{thm:gen-alpha-beta}. Let $S_1$ and $S_3$ be defined as in the statement of Theorem \ref{thm:gen-alpha-beta}. The set $S_1$ consists of all candidates that were selected to be in the candidate set of $M_1$ but are not in the candidate set of $M^*$, while $S_3$ is defined as the set of candidates that were not selected to be in the candidate set of $M_1$ but are in the candidate set of $M_1$. Thus by Lemma \ref{lem:loser-delta-dist} we have that $\sum_v \left[d(v, c) + d(c,p)\right] \leq \delta_{itr}\sum_v \left[d(v, c') + d(c',p)\right]$ for $c \in S_1$ and $c' \in S_3$ so $\alpha = \beta = \delta_{itr}$. Thus we have the distortion is bounded above by $2\deltaitr + 1$.
\end{proof}

So now we can see that using plurality veto is actually better than $9$. In fact, we can improve the upper bound on the distortion to $7$. 

\begin{cor}
    Iterative election with plurality veto gives a distortion of at most 7.
\end{cor}

To finish the proof of \ref{thm:2b1-dist-knowing-preferences} we will prove a more much general statement that will also be crucial to proving \ref{thm:serial_dictatorship_5}. To do this, we will  have to be significantly more careful about how we transform $M_1$ into $M_2$, so now we will construct $M_2$ in the following lemma and prove some properties of $M_2$.

\begin{lem}\label{lem:match-properties}
    Let $M_1$ be the matching returned by the iterative election procedure, and $M^*$ be the matching with minimum social cost. Let $S_1 = \cand(M_1) \setminus  \cand(M^*)$, $S_2 = \cand(M_1) \cap  \cand(M^*)$, and $S_3 =  \cand(M^*) \setminus \cand(M_1)$. Define the matching $M_2$ as follows:
    \begin{enumerate}
        \item For all positions $p$ with $M_1(p) \in S_2$ assign $M_2(p) = M_1(p)$. \label{step:step1}
        \item For all positions $p$ with $M_1(p) \in S_1$ and $M^*(p) \in S_3$ assign $M_2(p) = M^*(p)$. \label{step:step2}
        \item \label{step:step3} For any positions $p$ that is not assigned by step \ref{step:step1} or step \ref{step:step2} assign $M_2(p) = c$ for some arbitrary unassigned candidate $c \in S_3$. 
    \end{enumerate}
    Define the function $f(c) = M^*(\inv{M_2}(c))$; then we have that
    \begin{enumerate}
        \item $\cand(M_2) = S_2 \cup S_3 = \cand(M^*)$. \label{prop:prop1}
        \item The function $f \colon {S_2 \cup S_3} \to {S_2 \cup S_3}$ is bijective. \label{prop:prop2}
        \item If $c \in S_3$ then $f(c) \in S_2$ or $f(c) = c$. \label{prop:prop3}
    \end{enumerate}

\end{lem}
    \begin{proof}
    First we will show that $M_2$ is indeed a matching and $\cand(M_2) = S_2 \cup S_3 = \cand(M^*)$. Consider some position $p$ then clearly $p$ is only assigned once and is assigned to a candidate in either $S_2$ or $S_3$. A candidate in $S_2$ can only be assigned once since in step \ref{step:step1} we set $M_2(p) = M_1(p) = c$ and $c$ is not assigned in step \ref{step:step2} or \ref{step:step3}. If $c \in S_3$ then either $c$ is assigned only to its optimal positions or otherwise we make sure we only assign $c$ in step \ref{step:step3} to one position. We have that 
    \begin{align*}
        \abs{S_1} &= \abs{\cand(M_1) \setminus \cand(M^*)}\\ &= \abs{\cand(M_1)} - \abs{\cand(M_1) \cap \cand(M^*)}\\ &= \abs{\cand(M^*)} - \abs{\cand(M_1) \cap \cand(M^*)}\\ &= \abs{\cand(M^*) \setminus \cand(M_1)}\\ &= \abs{S_3}
    \end{align*}
    thus the algorithm terminates after each candidate in $S_1$ is swapped with a candidate in $S_3$ with no candidates remaining, so $\cand(M_2) = S_2 \cup S_3 = \cand(M^*)$.

    The function $f = M^*(\inv{M_2}(c))$ is clearly bijective since $M^*$ is a bijection and so is $M_2$ when restricted to the set ${S_2 \cup S_3}$. Now we will show that last property, that for $c \in S_3$ we have $f(c) = c$ or $f(c) \in S_2$. First suppose that $c \in S_3$ and $c$ was matched on step \ref{step:step2} then let $p = \inv{M^*}(c)$ so $M_2(p) = c$. Thus $f(c) = M^*(\inv{M_2}(c)) = M^*(p) = c$. Now suppose $c$ was matched in \ref{step:step3} then let $p = \inv{M_2}(c)$ so $M^*(p) \not \in S_3$ or otherwise $p$ would have been matched on step \ref{step:step2}. Thus we can conclude $f(c) = M^*(\inv{M_2}(c)) = M^*(p) \not \in S_3$ so $M^*(p) \in S_2$. Thus we have completed the proof of property \ref{prop:prop3}.
    \end{proof}
    
In other words, there exists a matching such that everything in $S_1$ has its assignment changed to the corresponding position in $M^*$ if possible, ensuring that two candidates in $S_3$ with never need to swap positions in these two matchings. Just as in the proof of Lemma \ref{lem:three-match-candate-set}, this allows us to transform matching $M_1$ into a matching with the same candidate set as $M^*$. With the properties shown above, we can then obtain the following theorem. 

\begin{thm} \label{thm:gen-alpha-beta}
     Suppose some mechanism gives a matching $M_1$ and assume $M^*$ is the optimal matching. Let $S_1 = \cand(M_1) \setminus  \cand(M^*)$, $S_2 = \cand(M_1) \cap  \cand(M^*)$, and $S_3 =  \cand(M^*) \setminus \cand(M_1)$. If  $\sum_{v\in \V} (d(v,c) + d(c,\inv{M_1}(c))) \leq \alpha \sum_{v\in \V} d(v,c') + \beta \sum_{v\in \V}d(c',\inv{M_1}(c))$ for every $c \in S_1$ and $c' \in S_3$, and with $\alpha + \beta \geq 2$, then $\cost(M_1) \leq (\alpha+\beta+1)\cost(M^*)$. \label{thm:gen-matching}
\end{thm}
\begin{proof}

   
    See Figure \ref{example-3-4-7} for an illustration of the proof below. As we have seen in Theorem \ref{thm:3-delta}, we can get a constant upper bound of $3\deltaitr$. However, by more carefully choosing $M_2$ in the proof for Theorem \ref{thm:3-delta}, we can obtain a better bound. To do this, first we start by defining some useful notation.
    Let $S \subseteq \C$ and $M$ be a matching. Then, $$\cost(M|S) = \sum_{v \in \V} \sum_{p :M(p) \in S} \left[d(v,M(p)) + d(M(p),p)\right],$$ or equivalently
    $$\cost(M|S) = 
    \sum_{v \in \V} \sum_{c \in S \cap \cand(M)} \left[d(v,c) + d(c, \inv{M}(c))\right].
    $$ The values of $\cost_V(M|S)$ and $\cost_P(M|S)$ are defined in the same manner. 
    
    Let $M_2$ be the matching given by Lemma \ref{lem:match-properties} and partition $S_3$ into sets $S_3^* = \{c \in S_3:M^*(c) = M_2(c)\}$ and $S'_3 = S_3\setminus S_3^*$.
    Our cost then is given as follows,
    \begin{align*}
        \cost(M_1) &= \sum_{v \in \V} \sum_{p \in \Pos} d(v, M_1(p)) + d(M_1(p), p)\\
        &= \cost(M_1|S_1) + \cost(M_1|S_2)\\
        &= \cost_P(M_1|S_1) + \cost_V(M_1|S_1) + \cost(M_2|S_2)\\
        &= \sum_{v \in \V} \sum_{c \in S_1} d(c, \inv{M_1}(c)) + \sum_{v \in \V} \sum_{c \in S_1} d(v, c) + \cost(M_2|S_2)\\
    \end{align*}

    Now consider the value of $c' = M_2(\inv{M_1}(c))$ for $c \in S_1$. Then $c' \in S_3$ because in either step \ref{step:step2} or step \ref{step:step3} of Lemma \ref{lem:match-properties} we assign $\inv{M_1}(c)$ to some $c' \in S_3$. We also have by our initial assumption in the Theorem statement that $\sum_{v\in \V} (d(v,c) + d(c,\inv{M_1}(c))) \leq \alpha \sum_{v\in \V} d(v,c') + \beta \sum_{v\in \V} d(c',\inv{M_1}(c))$. Thus 
    \begin{align*}
        \cost(M_1) &\leq \beta \sum_{v \in \V} \sum_{c \in S_1} d(M_2(\inv{M_1}(c)), \inv{M_1}(c)) + \alpha \sum_{v \in \V} \sum_{c \in S_1} d(v, M_2(\inv{M_1}(c))) + \cost(M_2|S_2)\\
        &= \beta \sum_{v \in \V} \sum_{c \in S_1} d(M_2(\inv{M_1}(c)), \inv{M_2}(M_2(\inv{M_1}(c)))) + \alpha \sum_{v \in \V} \sum_{c \in S_1} d(v, M_2(\inv{M_1}(c))) + \cost(M_2|S_2).
    \end{align*}
    We also note that $M_2(\inv{M_1})$ is a bijection from $S_1$ to $S_3$. Thus, we can sum over $c' = M_2(\inv{M_1}(c))$ in $S_3$ instead of over $S_1$ and obtain
 
    
    \begin{align*}
        \cost(M_1) &\leq \beta \sum_{v \in \V} \sum_{c \in S_3} d(c, \inv{M_2}(c)) + \alpha \sum_{v \in \V} \sum_{c \in S_3} d(v, c) + \cost(M_2|S_2)\\
        &= \beta \cost_P(M_2|S_3) + \alpha \cost_V(M_2|S_3) + \cost(M_2|S_2)\\
        &= \beta \left(\cost_P(M_2|S_3^*) + \cost_P(M_2|S_3')\right) + \alpha \left(\cost_V(M_2|S_3^*) + \cost_V(M_2|S_3')\right)\\ &\hspace{.5in} + \cost(M_2|S_2)
    \end{align*}
    
    This step corresponds to going from Figure \ref{fig:thm2b1a} to Figure \ref{fig:thm2b1b}. 
    
    Let us define the 
    mapping $f$ as $f(c) = M^*(\inv{M_2}(c))$. In other words, $f$ takes candidates and returns the candidate that should be matched to the position $c$ is currently matched to. Thus we can proceed with the analysis below, letting $\gamma = \max\{\alpha, \beta\}$ and observing that $\cost(M_2|S_3^*) = \cost(M^*|S_3^*)$ since both matchings are identical on $S^*_3$. Thus we have,
    
   \begin{align*}
        \cost(M_1) &\leq \gamma\cost(M_2|S_3^*) + \alpha \cost_V(M_2|S_3') + \beta\cost_P(M_2|S_3') + \cost(M_2|S_2)\\
        &= \gamma\cost(M^*|S_3^*) + \alpha \cost_V(M_2|S_3') + \beta\cost_P(M_2|S_3') + \cost(M_2|S_2)\\
        &= \gamma\cost(M^*|S_3^*) + \alpha\sum_{v \in \V} \sum_{c \in S_3'} d(v,c) + \beta \sum_{v \in \V} \sum_{c \in S_3'} d(c,\inv{M_2}(c)) + \cost(M_2|S_2)\\
        &\leq \gamma\cost(M^*|S_3^*) + \alpha\sum_{v \in \V} \sum_{c \in S_3'} d(v,c) + \beta\sum_{v \in \V} \sum_{c \in S_3'} \left[d(c,v) + d(v, f(c)) + d(f(c),\inv{M_2}(c))\right]\\
        &\hspace{.5in}  + \cost(M_2|S_2)\\
        &= \gamma\cost(M^*|S_3^*) + (\alpha + \beta) \sum_{v \in \V} \sum_{c \in S_3'}d(v,c) + \beta \sum_{v \in \V} \sum_{c \in S_3'} \left[d(v, f(c)) + d(f(c), \inv{M_2}(c))\right]\\&\hspace{.5in} + \cost(M_2|S_2)\\
        \end{align*}
        Now we have changed the set of candidates appearing in the above expression to the correct set $\cand(M^*) = S_2 \cup S_3$. So $\cost_V(M_1)$ has been bounded but not $\cost_P(M_1)$ since the assignments are still not necessarily all the same as $M^*$. The previous step corresponds to going from Figure \ref{fig:thm2b1c} to Figure \ref{fig:thm2b1d}. We utilize the fact that $\inv{M^*}(f(c)) = \inv{M^*}(M^*(\inv{M_2}(c))) = \inv{M_2}(c)$ below. Continuing with the proof:
        \begin{align}
            \cost(M_1) &\leq \gamma\cost(M^*|S_3^*) + (\alpha + \beta) \sum_{v \in \V} \sum_{c \in S_3'}d(v,c) + \beta \sum_{v \in \V} \sum_{c \in S_3'} \left[d(v, f(c)) + d(f(c), \inv{M_2}(c))\right]\nonumber \\&\hspace{.5in} + \cost(M_2|S_2)\nonumber \\
            &= \gamma\cost(M^*|S_3^*) + (\alpha + \beta) \sum_{v \in \V} \sum_{c \in S_3'}d(v,c) + \beta \sum_{v \in \V} \sum_{c \in S_3'} \left[d(v, f(c)) + d(f(c), \inv{M^*}(f(c)))\right]\nonumber \\&\hspace{.5in} + \cost(M_2|S_2)\nonumber \\
            &\leq \gamma\cost(M^*|S_3^*) + (\alpha + \beta) \sum_{v \in \V} \sum_{c \in S_3'}d(v,c) + \beta\sum_{v \in \V} \sum_{c \in S_3'} \left[ d(v, f(c)) + d(f(c), \inv{M^*}(f(c)))\right]\nonumber \\
            &\hspace{.5in} + \sum_{v \in \V} \sum_{c \in S_2}\left[d(v,c) + d(c,\inv{M_2}(c)\right]\nonumber \\
            &\leq \gamma\cost(M^*|S_3^*) + (\alpha + \beta) \sum_{v \in \V} \sum_{c \in S_3'}d(v,c)\nonumber \\
            &\hspace{.5in} + \beta\sum_{v \in \V} \sum_{c \in S_3'} \left[ d(v, f(c)) + d(f(c), \inv{M^*}(f(c)))\right]\nonumber \\\nonumber \\
            &\hspace{.5in} + \sum_{v \in \V} \sum_{c \in S_2}\left[ d(v,c) + d(c,v) + d(v, f(c)) + d(f(c), \inv{M_2}(c))\right]\nonumber \\
            &= \gamma\cost(M^*|S_3^*) \label{sum:sum1}\\
            &\hspace{.5in} + (\alpha + \beta) \sum_{v \in \V} \sum_{c \in S_3'}d(v,c) \label{sum:sum2}\\
            &\hspace{.5in} + \beta\sum_{v \in \V} \sum_{c \in S_3'} \left( d(v, f(c)) + d(f(c), \inv{M^*}(f(c)))\right) \label{sum:sum3}\\
            &\hspace{.5in} + \sum_{v \in \V} \sum_{c \in S_2} \left[2d(v,c) + d(v, f(c)) + d(f(c), \inv{M^*}(f(c)))\right]&\hspace{.5in} \label{sum:sum4}
        \end{align}

    To prove the theorem, our goal is to show that the above expression is at most $(\alpha+\beta +1)\cost(M^*)$. Note that $\cost(M^*)$ is simply the sum of $\sum_{v\in\V}[d(v,c)+ d(c,\inv{M^*}(c))]$ for all $c\in S_2\cup S_3$. So to complete the proof, we only need to show that if we consider only the contribution of each candidate $c$ separately to the above expression, that it is at most $(\alpha+\beta+1)$ times larger than $\cost(M^*|c) = \sum_{v\in\V}[d(v,c)+ d(c,\inv{M^*}(c))]$. To do this, we will now look at each candidate individually. 

    First consider $c \in S_3^*$. $c$ occurs only once in $\cost(M^*|S_3^*)$ on line \ref{sum:sum1} and no where else. Thus its contribution to the above sum is at most $\gamma \cost(M^*|c)$.

    Now consider any $c \in S_3'$. By Lemma \ref{lem:match-properties}, it must be that $c=f(c')$ for some $c'\in S_2$, since $f$ is a bijection and since all members of $S_3'$ are mapped to $S_2$ by $f$. First consider the term $d(v,c)$: this term occurs $(\alpha + \beta)$ times in line \ref{sum:sum2} and at most once in line \ref{sum:sum4} since $c =f(c')$ for $c'\in S_2$. Thus $d(v,c)$ occurs at most $\alpha + \beta + 1$ times for $c \in S_3'$. Now consider the term $d(c,\inv{M^*}(c))$ for $c \in S_3'$: this term can occur at most once in line \ref{sum:sum4} since $c \in f(S_2)$. Thus $c$'s contribution to the above sum is at most $(\alpha+\beta +1)\cost(M^*|c)$.

    Next consider any $c \in S_2$ such that $\inv{f}(c) \in S_2$. First we will look at $d(v,c)$: this term cannot occur on lines \ref{sum:sum1}, \ref{sum:sum2} or \ref{sum:sum3} since $\inv{f}(c) \in S_2$, but can occur at most 3 times in line \ref{sum:sum4}: twice as $d(v,c)$ and once as $d(v,f(c'))$ for $c'=\inv{f}(c)$. Now consider the term $d(c, \inv{M^*}(c))$: again this cannot occur in lines \ref{sum:sum1}, \ref{sum:sum2} or \ref{sum:sum3} but it can occur at most once in line \ref{sum:sum4}. Thus its contribution to the above sum is at most $3\cost(M^*|c)$.

    Lastly, consider $c \in S_2$ such that $\inv{f}(c) \in S_3'$. The term $d(v,c)$ occurs in line \ref{sum:sum3} with $c=f(c')$ for some $c'\in S_3'$, and occurs two more times in line \ref{sum:sum4}, for a total contribution of $\beta+2$. Likewise the term $d(c, \inv{M^*}(c))$ occurs only $\beta$ times in line \ref{sum:sum3}. Thus its total contribution to the above sum is at most $\max\{\beta+2,3\}\cost(M^*|c)$.
    



    Thus, since $\alpha+\beta \geq 2$,  the previous sum can be bounded by $(\alpha+\beta+1)\sum_{c\in cand(M^*)} \cost(M^*|c) = (\alpha+\beta+1)\cost(M^*)$, as desired.
\end{proof}

 In the above proof we bound out matching $M_1$ by the cost of a new matching $M_2$ which has the same candidate set as the optimal set, $\cand(M^*)$. Then, we further correct the assignments of positions within that set using the triangle inequality. Candidates in $S_3$ will never swap with each other so we can avoid the $3 \delta_{itr}$ that we encountered in Theorem \ref{thm:3-delta}.


\begin{figure}[p]
    \begin{subfigure}[b]{\linewidth}
    \centering
    \includegraphics[width=.8\textwidth]{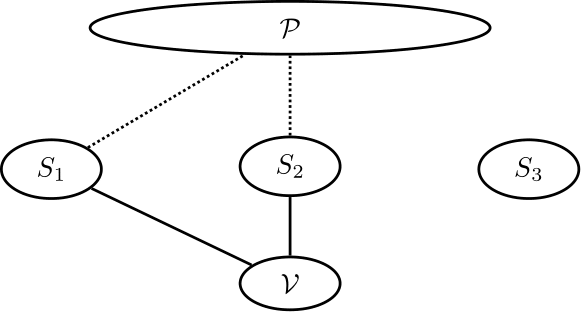}
    \caption{Our initial matching $M_1$ given by iterative election. The dotted lines indicate distances, $d(c,p)$, where it may not be the case that $M^*(p) = c$.}\label{fig:thm2b1a}
    \vspace{4ex}
  \end{subfigure}
  \begin{subfigure}[b]{\linewidth}
    \centering
    \includegraphics[width=.8\textwidth]{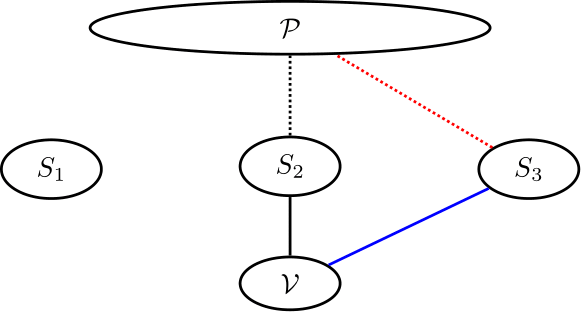}
    \caption{We now use the initial assumption of the theorem statement to convert the $\cost(M_1(S_1))$ to $\alpha \cost_P(M_2(S_3)) + \beta \cost_V(M_2(S_3))$. The cost associated with $d(c,p)$ is still not the same as in $M^*$ as candidates are not necessarily matched to the appropriate position. In this figure, costs with factor $\alpha$ appear in red, and those with factor $\beta$ appear in blue.}\label{fig:thm2b1b}
    \vspace{4ex}
  \end{subfigure}
    \end{figure}
  \begin{figure}[p]
    \ContinuedFloat
  \begin{subfigure}[b]
  {\linewidth}
    \centering
    \includegraphics[width=.8\textwidth]{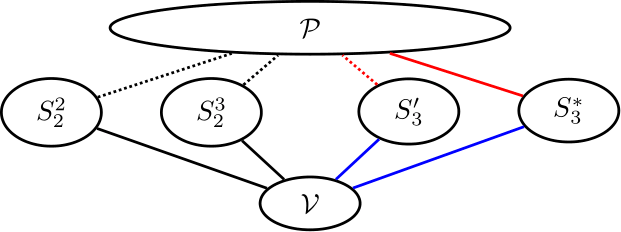}
    \caption{We now split the sets $S_3$ into $S_3'$ and $S_3^*$, and $S_2$ into $S_2^2$ and $S_3^3$. Distances from $S_3^*$ are drawn with a solid line, since those candidates are matched to the same position as in $M^*$. Define the set $S_2^2 = \{c \in S_2 : \inv{f}(c) \in S_2\}$ and $S_2^3 = \{c \in S_2 : \inv{f}(c) \in S_3\}$.}\label{fig:thm2b1c}
    \vspace{4ex}
  \end{subfigure}
  \begin{subfigure}[b]
  {\linewidth}
    \centering
    \includegraphics[width=.8\textwidth]{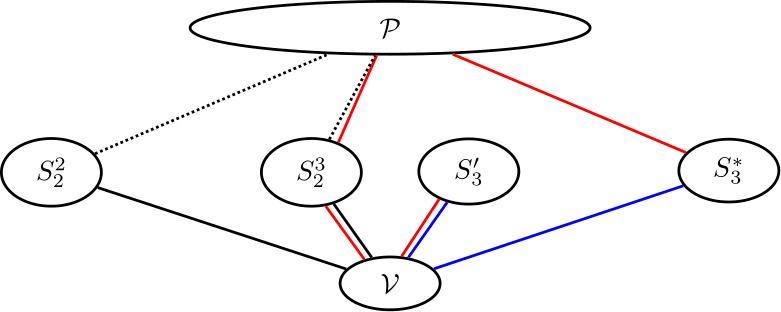}
    \caption{We can apply the triangle inequality to the positions incorrectly matched to candidates in $S_3'$. By the properties of $S_3'$ and $f$ these positions must be matched to candidates in $S_2^3$ in $M^*$.}\label{fig:thm2b1d}
    \vspace{4ex}
  \end{subfigure}
  \ContinuedFloat
  \begin{subfigure}[b]
  {\linewidth}
    \centering
    \includegraphics[width=.8\textwidth]{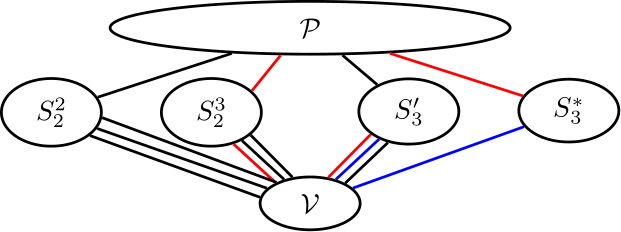}
    \caption{We have now shown that the cost of the original matching is at most the cost of various $M^*$ components. The component which appears the most is the distance from voters to $S_3'$: this appears at most $\alpha+\beta+1$ times.}\label{fig:thm2b1e}
    \vspace{4ex}
  \end{subfigure}
  \caption{Visual Proof of Theorem \ref{thm:gen-alpha-beta}.}
  \label{example-3-4-7} 
\end{figure}

\section{Knowing Ordinal Position Preferences}\label{sec:2}
Now that we have discussed the scenario where we only know the preferences of voters for each candidate, we look into the case where we know the ordinal position preferences for the candidates. In other words, for each position $p$ and candidates $c_1$ and $c_2$, we know which of $d(c_1,p)$ or $d(c_2,p)$ is smaller. So we only know the relative fitness of candidates for each position, but nothing about the voter preferences.


\begin{thm} \label{thm:k1serialdictator}
    There exists a mechanism that achieves a distortion of $3$ for \lineups{m}{1} while knowing only the position preferences. Moreover, no deterministic mechanism can have smaller worst-case distortion while knowing only the position preferences, even on a line and with only 2 candidates.
\end{thm}


\begin{proof}
    Let $\Pos = \{1\}$ and $c$ be the candidate that is the top choice of position $1$ and $c^*$ be the optimal candidate, let $M_1(1) = c$. Then
    \begin{align*}
        \cost(M_1) &= \sum_{v \in \V} \left[d(v,c) + d(c,1)\right]\\
        &\leq \sum_{v \in \V} \left[d(v,c^*) + d(c^*, 1) + d(c,1) + d(c,1)\right]\\
        &\leq \sum_{v \in \V} \left[d(v,c^*) + d(c^*, 1) + d(c^*,1) + d(c^*,1)\right]\\
        &= \sum_{v \in \V} \left[d(v,c^*) + 3d(c^*, 1)\right]\\
        &\leq 3\cost(c^*),
    \end{align*}
    as desired. To prove the lower bound on distortion, we will now consider two instances both with $\Pos = \{1\}$, $\C = \{A,B\}$ and $d(A,1) = 1-\epsilon$, $d(B,1) = 1+\epsilon$ and $d(A,B) = 2$. One instance will have $n$ voters on $A$ and the other will have $n$ on top of $B$. Since we have no information about the voters, these instances are equivalent in the eyes of any mechanism which only knows position preferences. If we chose $A$ then the second instance will have a distortion of $\lim_{\epsilon \to 0}\frac{n (d(A,B) + d(A,1))}{n d(B,1)} = \lim_{\epsilon \to 0}\frac{3-\epsilon}{1+\epsilon} = 3$. If we pick $B$ then the first instance will have a distortion of $\lim_{\epsilon \to 0}\frac{n (d(A,B) + d(B,1))}{n d(A,1)} = \lim_{\epsilon \to 0}\frac{3+\epsilon}{1-\epsilon} = 3$.
\end{proof} 

Note that the above mechanism is simply just picking the closest candidate to the position. For the rest of this section, we will be using the more general {\em serial dictatorship} mechanism (see e.g., \citet{biro2017applications,caragiannis2022truthful}): for each position $p$ in some arbitrary order, we assign $M(p)$ to be the top choice of $p$ among remaining unassigned candidates. 

After considering the case where we have $m$ candidates and one position, we now extend the setting to $m$ candidates and $m$ positions. In this case, each candidate will be matched to one position and vice versa.
We will first show that serial dictatorship gives a trade off between the cost of the position side of the matching and the cost of the voter side of the matching. To do that we will revisit Lemma \ref{lem:all-perfect-matching-3} and derive a new lemma.

\begin{lem}
    For all matchings $M_1$, $M_2$ if $\abs{\C} = \abs{\Pos}$, then $\cost(M_1) \leq 3\cost_V(M_2) + \cost_P(M_2)$.
\end{lem}
\begin{proof}
    This statement follows directly from Lemma \ref{lem:all-perfect-matching-3} line \ref{perfect-mtch-3vm}.
\end{proof}

Notice that on the right hand side, only the cost of the distances from the voters to the candidates is increased by a factor of 3 while the cost of the candidate-position matching remains unchanged. We then try optimizing the matching and find that we have a trade off. As discussed below, we will utilize this trade off to obtain an upper bound of the distortion smaller than 3. 

\begin{thm} \label{thm:kkserial}
    Serial dictatorship gives a distortion of at most $3 - \frac{1}{2^{m-2}}$ for \lineups{m}{m}.
\end{thm}

\begin{proof}
    Let $M_1$ be the matching given by serial dictatorship and $M^*$ be the optimal matching. We know from Theorem 1 of \citet{caragiannis2022truthful} that serial dictatorship for bipartite perfect matching gives a distortion of at most $2^m - 1$ so we have $\cost_P(M_1) \leq (2^m - 1)\cost_P(M^*)$. Therefore, due to the Lemma above, and the fact that for \lineups{m}{m} we have $\cost_V(M_1)=\cost_V(M^*)$, our distortion is given by
    \begin{align*}
        \min \left \{ \frac{3\cost_V(M^*) + \cost_P(M^*)}{\cost_V(M^*) + \cost_P(M^*)}, \frac{\cost_V(M^*) + (2^m-1)\cost_P(M^*)}{\cost_V(M^*) + \cost_P(M^*)}\right\}.
    \end{align*}
    This is maximized when $\frac{3\cost_V(M^*) + \cost_P(M^*)}{\cost_V(M^*) + \cost_P(M^*)} =  \frac{\cost_V(M^*) + (2^k-1)\cost_P(M^*)}{\cost_V(M^*) + \cost_P(M^*)}$. So we have
    \begin{align*}
        3\cost_V(M^*) + \cost_P(M^*) &= \cost_V(M^*) + (2^k-1)\cost_P(M^*)\\
        \cost_V(M^*) &= (2^{k-1}-1)\cost_P(M^*).
    \end{align*}
    Thus $\frac{\cost_P(M_1) + \cost_V(M_1)}{\cost_P(M^*) + \cost_V(M^*)} \leq \frac{\cost_P(M^*) + 3(2^{m-1}-1)\cost_P(M^*)}{\cost_P(M^*) + (2^{m-1}-1)\cost_P(M^*)} = \frac{3 \cdot 2^{m-1}-2}{2^{m-1}} = 3- \frac{1}{2^{m-2}}$.
\end{proof}

Theorem \ref{thm:kkserial} shows that simply by using serial dictatorship we can get a distortion better than 3. 
For $m=2$ we get a distortion of $2$, and for $m=3$ we get $\frac{5}{2}$, so for small $m$ the distortion remains relatively low despite approaching 3 as $m$ becomes larger. 

We now consider the general case where we have $m$ candidates and $\ell$ positions, with $m\geq \ell$. We will now show that simply running serial dictatorship on the set of positions results in an upper bound on the distortion of $5$. Having the extra information on the ordering of candidates by fitness for each position not only allows us to improve the bounds from Theorem \ref{thm:2b1-dist-knowing-preferences} when we only know voter preferences, but also does so without needing voter preferences at all. Note that serial dictatorship then is also strategy proof, since we do not consider the preferences of the candidates or the voters, and even the positions are not incentivized to lie (if such a thing were possible) to achieve a closer candidate.

\begin{thm}
    Serial dictatorship achieves a distortion of at most $5$ for \lineups{m}{\ell}. \label{thm:serial_dictatorship_5}
\end{thm}

\begin{proof}
    Suppose serial dictatorship returns a matching $M_1$, using the same definitions as we have before, let $S_1 = \cand(M_1) \setminus  \cand(M^*)$, $S_2 = \cand(M_1) \cap  \cand(M^*)$, and $S_3 =  \cand(M^*) \setminus \cand(M_1)$. Take $c \in S_1$ and $c' \in S_3$. Then the contribution of $c$ to the cost of $M_1$ is:

    \begin{align*}
        \sum_{v \in \V} [d(v,c) + d(c,\inv{M_1}(c))]
        \leq \sum_{v \in \V} [d(v,c') + d(c',\inv{M_1}(c)) + d(c,\inv{M_1}(c)) + d(c,\inv{M_1}(c))].
    \end{align*}
    Since $c'$ was not chosen when $c$ was chosen for position $M_1^{-1}(c)$, then $d(c,\inv{M_1}(c)) \leq d(c',\inv{M_1}(c))$. Therefore the above quantity is at most,
    \begin{align*}
        &\leq \sum_{v \in \V} [d(v,c') + d(c',\inv{M_1}(c)) + d(c',\inv{M_1}(c)) + d(c',\inv{M_1}(c))]\\
        &\leq \sum_{v \in \V} [d(v,c') + 3d(c',\inv{M_1}(c))]\\
        &= \sum_{v \in \V} d(v,c') + 3\sum_{v \in \V} d(c',\inv{M_1}(c)).
    \end{align*}
    We then have the necessary conditions to apply Theorem \ref{thm:gen-matching}. We have values $\alpha=1$ and $\beta = 3$ so the distortion is at most $\alpha + \beta + 1 = 5$.
\end{proof}

We will now briefly consider the case for when we also know the voter ordinal preferences for the candidates in addition to the ordinal position preferences. We show that we can improve the upper and lower bounds from $3$ to $2$ for \lineups{2}{1}.


\begin{thm} \label{thm:21serialdictator}
    Consider a \lineup{2}{1} knowing position preference information {\em and} voter preference information, with $\Pos = \{1\}$ and $\C = \{A,B\}$. Without loss of generality assume we have $d(A,1) \leq d(B,1)$ and there are $n_a$ voters who prefer $A$ to $B$ and $n_b$ who prefer $B$ to $A$. Consider the mechanism which selects candidate $A$ if $n_b \leq 2n_a$, otherwise select candidate $B$. Then, this mechanism achieves a distortion of at most $2$. Moreover, no deterministic mechanism can do better, even on a line.
\end{thm}

\begin{proof}
       Suppose $\C = \{A,B\}$ and $\Pos = \{1\}$. First we will show a reduction that will allow us to consider only a subset of all line-up elections. We do this by showing that for every line-up election with 2 candidate and 1 position, there exists an embedding of points in $\R$ with distortion at least as large. Let $A'$ be at position $0$ in $\R$ and $B'$ at position $d(A,B)$, and define $\zeta(x) = \frac{1}{2}\left[d(x,A) + d(x,B) - d(A,B)\right]$. Now place voter $j'$ at $d(j,A) - \zeta(j)$ and position $1$ at $d(A,1) - \zeta(1)$. 

       Note that every $j'$ lies between $A'$ and $B'$:
       
       \begin{align*}
           d(j',A') &= d(j,A) - \zeta(j)\\
           &= d(j,A) - \frac{1}{2}\left[d(j,A) + d(j,B) - d(A,B)\right]\\
           &= \frac{1}{2}d(j,A) - \frac{1}{2}d(j,B) + \frac{1}{2}d(A,B)\\
           &\geq \frac{1}{2}d(j,B) - \frac{1}{2}d(j,B)\\
           &= 0.
       \end{align*}
       We also have
       \begin{align*}
           d(j',A') &= d(j,A) - \zeta(j)\\
           &= d(j,A) - \frac{1}{2}\left[d(j,A) + d(j,B) - d(A,B)\right]\\
           &= \frac{1}{2}d(j,A) - \frac{1}{2}d(j,B) + \frac{1}{2}d(A,B)\\
           &\leq \frac{1}{2}d(A,B) + \frac{1}{2}d(A,B)\\
           &= d(A,B).
       \end{align*}
       Thus $j'$ lies in between $A'$ and $B'$ since $j'$ is to the right of $A'$ and must be to the left of $B'$, and the same holds for $1$. We also have that 
       \begin{align*}
            d(j',B') &= d(A',B') - d(j', A')\\
            &= d(A,B) - d(j', A')\\
            &= d(A,B) - d(j, A) + \zeta(j)\\
            &= d(A,B) - d(j, A) + \frac{1}{2}\left[d(j,A) + d(j,B) - d(A,B)\right]\\
            &= d(j,B) - \frac{1}{2}\left[d(j,A) + d(j,B) - d(A,B)\right]\\
            &= d(j,B) - \zeta(j).
       \end{align*} and the analogous result $d(B',1') = d(B, 1) - \zeta(1)$ holds for 1.
       
        Now we will show that the voter (and position) preferences remain unchanged, as compared with the original instance. Suppose $j$ prefers $A$. Then,

       \begin{align*}
           d(j',A') + d(A',1') &= d(j,A) + d(A,1) - \zeta(j) - \zeta(1)\\
           &\leq d(j,B) + d(B,1) - \zeta(j) - \zeta(1)\\
           &= d(j',A') + d(A',1)
       \end{align*}

       so $j'$ also prefers $A'$. Similarly, if $j$ prefers $B$ then so does $j'$ prefer $B'$.
        Thus $j'$ has the same preferences as $j$. Now we will show the distortion is no smaller with this embedding. Suppose $M_1$ is the matching selecting $A$, and $M^*$ is the optimal matching selecting $B$, without loss of generality. We have that the distortion is given by

       \begin{align*}
        \frac{\cost(M_1)}{\cost(M^*)} &= \frac{\sum_{v \in \V} \left[d(v,A) + d(A,1)\right]}{\sum_{v \in \V} \left[d(v,B) + d(B,1)\right]}\\
        &= \frac{\sum_{v \in \V} \left[d(v',A') + d(A',1') + \zeta(v) + \zeta(1)\right]}{\sum_{v \in \V} \left[d(v',B') + d(B',1') + \zeta(v) + \zeta(1)\right]}\\
        &\leq \frac{\sum_{v \in \V} \left[d(v',A') + d(A',1')\right]}{\sum_{v \in \V} \left[d(v',B') + d(B',1')\right]}.
    \end{align*}
    The last line is simply the distortion in our new instance in $\R$. Thus when bounding the worst-case distortion of our mechanism we need only consider instances on a line, and moreover only instances where voters lie between the two candidates. From this point on in the proof, we will analyze the worst-case distortion of our mechanism, but we will assume that $d(j,A) + d(j,B) = d(A,B)$ for all $j$.

         Let $\V(A)$ be the set of all voters that prefer $A$ to $B$ and $\V(B)$ be the set of all voters that prefer $B$ (breaking ties arbitrarily). 
        
        First suppose we pick $A$ for our matching $M_1(1) = A$. Let $f \colon \V(B) \to \V(A)$ such that if $f(i_1) = f(i_2)$ then either $i_1 = i_2$ or there does not exist a $i_3 \not = i_1$ and $i_3 \not = i_2$ such that $f(i_3) = f(i_1) = f(i_2)$, in other words $f$ maps voters from $\V(B)$ to voters in $\V(A)$ such that at most two voters in $\V(B)$ are mapped to the same voter in $\V(A)$. Then we have
    \begin{align*}
        \cost(M_1) &= \sum_{i \in \V} [d(i,A) + d(A,1)]\\
        &= \sum_{i \in \V(A)} \left[d(i,A) + d(A,1)\right] + \sum_{i \in \V(B)} \left[d(i,A) + d(A,1)\right]\\
        &\leq \sum_{i \in \V(A)} \left[d(i,B) + d(B,1)\right] + \sum_{i \in \V(B)} \left[d(i,A) + d(A,1)\right]\\
        &\leq \sum_{i \in \V(A)} \left[d(i,B) + d(B,1)\right] \\ &\hspace{.25in}+\sum_{i \in \V(B)} \left(\frac{1}{2} \left[d(i,B) + d(B,1) + d(A,1) + d(i,B) + d(f(i), B) + d(f(i), A)\right] + d(A,1)\right).
    \end{align*}

    The inequalities above are due to all voters in $\V(A)$ preferring $A$, and the triangle inequality. Now consider the cost of a single voter $i \in \V(B)$ recalling that $d(j,A) + d(j,B) = d(A,B)$ for $j \in \V(A)$.
    \begin{align*}
        \cost(i) &\leq \frac{1}{2} \left[d(i,B) + d(B,1) + d(A,1) + d(i,B) + d(f(i), B) + d(f(i), A)\right] + d(A,1)\\
        &= d(i,B) + \frac{1}{2} \left[d(B,1) + d(A,1) + d(A,B)\right] + d(A,1)\\
        &\leq d(i,B) + \frac{1}{2} \left[d(B,1) + d(A,1) + \frac{1}{2}\left[d(A,1) + d(B,1) + d(f(i), A) + d(f(i), B)\right]\right] + d(A,1)\\
        &\leq d(i,B) + \frac{1}{2} \left[d(B,1) + d(B,1) + \frac{1}{2}\left[d(B,1) + d(B,1) + d(f(i), B) + d(f(i), B)\right]\right] + d(B,1)\\
        &= d(i,B) + \frac{1}{2} \left[2d(B,1) + \frac{1}{2}\left[2d(B,1) + 2d(f(i), B)\right]\right] + d(B,1)\\
        &= d(i,B) + 2d(B,1) + \frac{1}{2}\left[d(f(i), B) + d(B,1)\right]
    \end{align*}
    The above inequalities are due to the triangle inequality, the fact that $d(A,1)\leq d(B,1)$, and the fact that $f(i)\in \V(A)$.
    Thus we can conclude
    \begin{align*}
        \cost(M_1) 
        &\leq \sum_{i \in \V(A)} \left[d(i,B) + d(B,1)\right] \\ &\hspace{.25in}+\sum_{i \in \V(B)} \left(\frac{1}{2} \left[d(i,B) + d(B,1) + d(A,1) + d(i,B) + d(f(i), B) + d(f(i), A)\right] + d(A,1)\right)\\
        &\leq \sum_{i \in \V(A)} \left[d(i,B) + d(B,1)\right] +\sum_{i \in \V(B)} \left(d(i,B) + 2d(B,1) + \frac{1}{2}\left[d(f(i), B) + d(B,1)\right]\right)\\
        &\leq \sum_{i \in \V(A)} \left[2d(i,B) + 2d(B,1)\right] +\sum_{i \in \V(B)} \left(d(i,B) + 2d(B,1)\right)
    \end{align*}

    Where the last line follows from each $f(i) \in \V(A)$ occurring at most twice in the second summation. Thus the cost of $M_1$ is at most twice that of optimum when we choose $A$.
    
    Now suppose our mechanism chooses $B$ (i.e., $M_1(1)=B$) and $A$ is optimal so $M^*(1)=A$. This means that $n_b > 2n_a$. Let $f_1 \colon \V(A) \to \V(B)$ and $f_2 \colon \V(A) \to \V(B)$ both be one-to-one functions such that there does not exist a $i_1,i_2 \in V(A)$ such that $f_1(i_1) = f_2(i_2)$; so $f_1$ and $f_2$ are injective functions such that the ranges do not intersect. In this case,
    \begin{align*}
        \cost(M_1) &= \sum_{i \in \V} \left[ d(i,B) + d(B,1)\right]\\
        &= \sum_{i \in \V(A)} \left[d(i,B) + d(B,1)\right] + \sum_{i \in \V(B)} \left[d(i,B) + d(B,1)\right]\\
        &\leq \sum_{i \in \V(A)} \left[d(i,B) + d(B,1)\right] + \sum_{i \in \V(B)} \left[d(i,A) + d(A,1)\right]\\
        &\leq \sum_{i \in \V(A)} \left(\frac{1}{2}\left[d(i,A) + d(A, f_1(i)) + d(f_1(i),B) + d(i,A) + d(A, f_2(i)) + d(f_2(i),B)\right] + d(B,1)\right)\\&\hspace{.25in} + \sum_{i \in \V(B)} \left[d(i,A) + d(A,1)\right]\\
        &= \sum_{i \in \V(A)} \left(\frac{1}{2}\left[d(i,A) + d(A, f_1(i)) + d(f_1(i),B) + 2d(B,1) + d(i,A) + d(A, f_2(i)) + d(f_2(i),B)\right]\right)\\&\hspace{.25in} + \sum_{i \in \V(B)} \left[d(i,A) + d(A,1)\right]\\
        &\leq \sum_{i \in \V(A)} \left(\frac{1}{2}\left[d(i,A) + d(A, f_1(i)) + d(f_1(i),A) + 2d(A,1) + d(i,A) + d(A, f_2(i)) + d(f_2(i),A) \right]\right)\\&\hspace{.25in} + \sum_{i \in \V(B)} \left[d(i,A) + d(A,1)\right]
        \end{align*}
        The above inequality is because $f_1(i)$ and $f_2(i)$ both prefer $B$ to $A$ so $d(f_1(i), B) + d(B,1) \leq d(f_1(i), A) + d(A,1)$ as well as $d(f_2(i), B) + d(B,1) \leq d(f_2(i), A) + d(A,1)$. Thus we can continue bounding the cost as
        \begin{align*}
        \cost(M_1) &\leq \sum_{i \in \V(A)} \left(\frac{1}{2}\left[d(i,A) + d(A, f_1(i)) + d(f_1(i),A) + d(i,A) + d(A, f_2(i)) + d(f_2(i),A)\right] + d(A,1)\right)\\&\hspace{.25in} + \sum_{i \in \V(B)} \left[d(i,A) + d(A,1)\right]\\
        &= \sum_{i \in \V(A)} \left(d(i,A) + d(A,1) + \frac{1}{2}\left[d(A, f_1(i)) + d(f_1(i),A) + d(A, f_2(i)) + d(f_2(i),A)\right]\right)\\&\hspace{.25in} + \sum_{i \in \V(B)} \left[d(i,A) + d(A,1)\right]\\
        &= \sum_{i \in \V(A)} \left(d(i,A) + d(A,1) + d(f_1(i),A)  + d(f_2(i),A)\right)\\&\hspace{.25in} + \sum_{i \in \V(B)} \left[d(i,A) + d(A,1)\right]\\
        &\leq \sum_{i \in \V(A)} \left[d(i,A) + d(A,1)\right] + \sum_{i \in \V(B)} \left[2d(i,A) + d(A,1)\right]\\
        &\leq 2\cost(M^*).
    \end{align*}

    \begin{figure}[ht]
    \centering
    \begin{minipage}{.5\textwidth}
      \centering
    \includegraphics[]{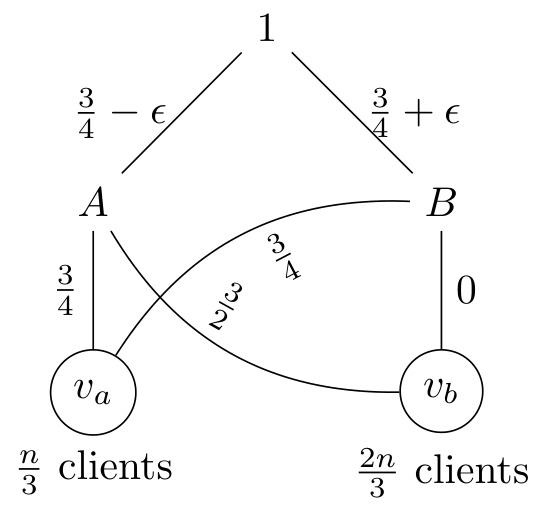}
    \subcaption{Metric space $d_1$.}
    \end{minipage}%
    \begin{minipage}{.5\textwidth}
      \centering
    \includegraphics[]{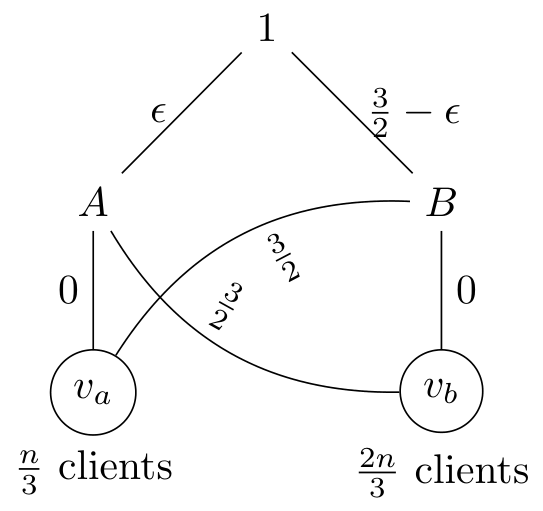}
    \subcaption{Metric space $d_2$.}
    \end{minipage}
    \caption{Two consistent metric spaces with distortion $2$.}
    \label{fig:thm41}
\end{figure}

    Thus we have a mechanism that achieves the best possible distortion of $2$. Now to show that no mechanism can do better we consider the two instances given in Figure \ref{fig:thm41}. Both have the same preference profiles, so our mechanisms cannot tell them apart. We can embed $d_1$ into $\R$ by setting $A$ to be at $0$, $1$ at $\frac{3}{4}-\epsilon$, $v_A$ at $\frac{3}{4}$ and both $B$ and $v_b$ at $\frac{3}{2}$. We can also embed $d_2$ into $\R$ by setting both $A$ and $v_a$ to be $0$, $1$ at $\epsilon$, and both $v_b$ and $B$ at $\frac{3}{2}$. If we pick $A$ based on the preference information and we are in metric space $d_1$, then we have $\cost(A) = \frac{n}{3}\left(\frac{3}{4} + \frac{3}{4} - \epsilon\right) + \frac{2n}{3}\left(\frac{3}{2} + \frac{3}{4}-\epsilon\right) = n(2-\epsilon)$ and $\cost(B)=\frac{n}{3}\left(\frac{3}{4} + \frac{3}{4} + \epsilon\right) + \frac{2n}{3}\left(\frac{3}{4}\right) = n(1 + \frac{1}{3}\epsilon)$ so $\lim_{\epsilon \to 0}\frac{\cost(A)}{\cost(B)} = 2$. Likewise if we chose $B$ then in metric space $d_2$ $\cost(A) = \frac{n}{3}\left(\epsilon\right) + \frac{2n}{3}\left(\frac{3}{2} + \epsilon\right) = n(\epsilon+1)$ and $\cost(B) = \frac{n}{3}\left(\frac{3}{2} + \frac{3}{2} - \epsilon\right) + \frac{2n}{3}\left(\frac{3}{2} - \epsilon\right) = n(2-\epsilon)$. Thus $\lim_{\epsilon \to 0}\frac{\cost(B)}{\cost(A)} = 2$. So we have a lower bound of $2$ as desired.
\end{proof}

\section{Knowing Candidate and Position Locations}\label{sec:3}

In this section, we look at what is possible when candidate and position locations are known exactly (and thus all the distances between them are known), but absolutely nothing is known about voter locations or preferences. This is often possible since information about candidates and positions can be public knowledge, but voter locations and preferences are either private, or surveying a large number of people about their preferences is impractical. 

\begin{thm}
    There exists a deterministic mechanism that gives a distortion of at most $3$ for the \lineup{m}{\ell} while knowing only candidate-position distance pairs. Moreover, no deterministic mechanism can achieve a distortion bound less than $3$ while knowing only position and candidate locations. 
\end{thm}

\begin{proof}
    Let $M$ be the matching that minimizes the cost of $\cost_P(M)$ and $M^*$ be the optimal matching. $M$ can be computed from the information we have. Then we have
    \begin{align*}
        \cost(M) &= \sum_{v \in \V} \sum_{p \in \Pos} \left[d(v, M(p)) + d(M(p), p)\right]\\
        &\leq \sum_{v \in \V} \sum_{p \in \Pos} \left[d(v, M^*(p)) + d(M^*(p), p) + d(M(p), p) + d(M(p), p)\right]\\
        &= \sum_{v \in \V} \sum_{p \in \Pos} \left[d(v, M^*(p)) + d(M^*(p), p)\right] + \sum_{v \in \V} \sum_{p \in \Pos} 2d(M(p), p)
    \end{align*}
    Since $M$ minimizes the cost of $\sum_{v \in \V} \sum_{p \in \Pos} d(M(p), p)$, we have
    \begin{align*}
        \cost(M) &\leq \sum_{v \in \V} \sum_{p \in \Pos} \left[d(v, M^*(p)) + d(M^*(p), p)\right] + \sum_{v \in \V} \sum_{p \in \Pos} 2d(M^*(p), p)\\
        &\leq \sum_{v \in \V} \sum_{p \in \Pos} \left[d(v, M^*(p)) + 3d(M^*(p), p)\right]\\
        &\leq 3 \cost(M^*).
    \end{align*}
    as desired. 

    To prove the lower bound, consider two instances both with $\Pos = \{1\}$, $\C = \{A,B\}$ and $d(A,1) = 1-\epsilon$, $d(B,1) = 1+\epsilon$ and $d(A,B) = 2$. One instance will have $n$ voters on $A$ and the other will have $n$ on top of $B$. If we choose $A$, then the second instance will have a distortion of $\lim_{\epsilon \to 0}\frac{n (d(A,B) + d(A,1))}{n d(B,1)} = \lim_{\epsilon \to 0}\frac{3-\epsilon}{1+\epsilon} = 3$. If we choose $B$ then the first instance will have a distortion of $\lim_{\epsilon \to 0}\frac{n (d(A,B) + d(B,1))}{n d(A,1)} = \lim_{\epsilon \to 0}\frac{3+\epsilon}{1-\epsilon} = 3$.
\end{proof}

Much like the mechanisms in Theorem \ref{thm:k1serialdictator} and \ref{thm:kkserial} this mechanism is strategyproof and does not rely on the preferences of the voters. This means that we can achieve low distortion without surveying voters, and instead focusing on increasing our knowledge of the pairwise distances between candidates and positions. Note that for \lineups{m}{m}, the above mechanism trivially computes the optimum solution, since all candidates must always be chosen in such elections, and thus $\cost_V(M)$ of any matching $M$ is always the same, and only $\cost_P(M)$ needs to be minimized.

We will now consider the case when we know the exact positions of candidates and positions {\em in addition} to voter ordinal preferences. We will consider the case where we have a single position, and first look at the case where there are 2 candidates.

The following is a lemma that will be useful later in the proof on Theorem \ref{thm:mech-1-k-dist}. 
    \begin{lem}
        The function $f(x,y) = 1 + \frac{2y}{(x+y) \frac{y}{x} + x} \leq \frac{5}{3}$ for $x,y > 0$. \label{lem:53-function}
    \end{lem}

    \begin{proof}
        First note that the function $f(\lambda x, \lambda y) = f(x,y)$ so $f$ is a homogeneous function of degree $0$. Thus we can set $x+y =1$ and consider the maximum of a function of a single variable $g(x) = 1 + \frac{2(1-x)}{\frac{1-x}{x} + x}$. This function is clearly continuous for $x \in (0,1)$ and has a single critical point at $x = \frac{1}{2}$ and is concave down with value $\frac{5}{3}$.
    \end{proof}

\begin{thm}
    For the \lineup{2}{1} given distances for candidate-position pairs, together with voter preferences, there exists a deterministic mechanism that achieves a distortion of at most $\frac{5}{3}$. \label{thm:mech-1-k-dist}
\end{thm}

\begin{proof}
Assume $\C = \{A,B\}$ and $\Pos = \{1\}$. Denote the number of voters who rank $A$ as their top choice by $n_a$, and the number of voters who rank $B$ as their top choice by $n_b$. The mechanism is as follows: we choose $A$ as the winner if $n_b\cdot d(A,1) \leq n_a\cdot d(B,1)$ and choose $B$ otherwise. Clearly, if $n_a = 0$ or $n_b = 0$, we can pick the optimal candidate accordingly, so we will assume $n_a > 0$ and $n_b > 0$.  Without loss of generality, assume we chose $A$ but choosing $B$ is optimal. Let $M_A$ be the matching such that $M_A(1) = A$ and $M_B$ be the matching such that $M_B(1) = B$. Denote the set of voters who prefer $A$ as $\V(A)$ and those who prefer $B$ as $\V(B)$. Also note that the optimal value is $\cost(M_B)$. First we can write the cost of $M_A$ as
    
     \begin{align*}
         \cost(M_A) &= \sum_{i\in \V} \left[d(i,A) + d(A,1)\right]\\
         &= \sum_{i\in \V} \left[d(i,A) + d(A,1) + d(i, B) + d(B,1) - d(i, B) - d(B, 1)\right]\\
         &= \sum_{i\in \V} \left[d(i, B) + d(B,1)\right] + \sum_{i\in \V}\left[d(i,A) + d(A,1) - d(i, B) - d(B, 1)\right]\\
         &= \sum_{i\in \V} \left[d(i, B) + d(B,1)\right] + \sum_{i\in \V(A)}\left[d(i,A) + d(A,1) - d(i, B) - d(B, 1)\right]\\
         &\hspace{.25in} + \sum_{i\in \V(B)}\left[d(i,A) + d(A,1) - d(i, B) - d(B, 1)\right]\\
         &=\cost(M_B) + \sum_{i\in \V(A)} \left[ d(i,A) - d(i,B) + d(A,1) - d(B,1)\right] +\\
         &\hspace{.25in} n_b(d(A,1) - d(B,1)) + \sum_{j \in \V(B)} (d(j,A) - d(j,B)).
     \end{align*}
     If $i \in \V(A)$ then $i$ prefers $A$ to $B$, so $$\sum_{i\in \V(A)} \left[ d(i,A) - d(i,B) + d(A,1) - d(B,1)\right] = \sum_{i\in \V(A)} \left[ (d(A,1)+d(A,i)) - (d(B,1) + d(B,i))\right] \leq 0.$$ Thus $\cost(M_A)$ can be bounded above by
     \begin{align*}
         \cost(M_A) &\leq \cost(M_B) + n_b(d(A,1) - d(B,1)) + \sum_{j \in \V(B)} (d(j,A) - d(j,B)).\\
     \end{align*}
     Here using triangle inequality we have that $\sum_{j \in \V(B)} (d(j,A) - d(j,B)) \leq n_b \cdot d(A,B)$. Therefore, we have
     \begin{align*}
         \cost(M_A) &\leq \cost(M_B) + n_b \left(d(A,1) - d(B,1) + d(A,B)\right)\\
         &\leq \cost(M_B) + n_b \left(d(A,1) - d(B,1) + d(A,1) + d(B,1)\right)\\
         &= \cost(M_B) + 2n_b \cdot d(A,1).
     \end{align*}
     In addition, we also have that
     \begin{align*}
          \cost(M_A) &\leq \cost(M_B) + n_b(d(A,1) - d(B,1)) + \sum_{j \in \V(B)} (d(j,A) - d(j,B))\\
          &\leq \cost(M_B) + n_b \left(d(A,1) - d(B,1) + d(A,B)\right)\\
          &\leq \cost(M_B) + n_b(d(A,1) - d(B,1) + d(A,i) + d(B,i)), \forall i \in \V(A)\\
          &= \cost(M_B) + n_b(d(A,1) - d(B,1) + d(A,i) + d(B,i) + d(B,i) - d(B,i)), \forall i \in \V(A)\\
          &\leq \cost(M_B) + n_b([d(A,1) + d(A,i)] - [d(B,1) + d(B,i)] + 2d(B,i)), \forall i \in \V(A).\\
     \end{align*}
     Here note that since $i$ prefers A, we have $[d(A,1) + d(A,i)] - [d(B,1) + d(B,i)] \leq 0$. Then we can see that 
     $$\cost(M_A) \leq \cost(M_B) + 2n_b\cdot d(B,i), \forall i \in \V(A).$$
     
    Let $\mu = \frac{1}{n_a} \sum_{i \in \V(A)} d(i,B)$. Then, since $\cost(M_A) \leq \cost(M_B) + 2n_b\cdot d(B,i)$ for all $i \in \V(A)$, we have  $\cost(M_A) \leq \cost(M_B) + 2n_b\cdot\mu$. Thus the distortion is bounded by
     \begin{align*}
         \frac{\cost(M_A)}{\cost(M_B)} \leq 1 + \min \left\{\frac{2n_b\cdot d(A,1)}{\cost(M_B)}, \frac{2n_b\cdot \mu}{\cost(M_B)}\right\}.
     \end{align*}
     There are two cases: (i) $\mu \geq d(A,1)$ and (ii) $\mu \leq d(A,1)$. We first consider case (i), $\mu \geq d(A,1)$. Then we have 
     \begin{align*}
         \frac{\cost(M_A)}{\cost(M_B)} &\leq 1 + \min \left\{\frac{2n_b\cdot d(A,1)}{\cost(M_B)}, \frac{2n_b\cdot \mu}{\cost(M_B)}\right\}\\
         &= 1 + \frac{2n_b\cdot d(A,1)}{\cost(M_B)}\\
         &= 1 + \frac{2n_b\cdot d(A,1)}{(n_a + n_b) d(B,1) + \sum_{i \in \V(A)} d(i, B) + \sum_{j \in \V(B)} d(j, B)}\\
         &\leq 1 + \frac{2n_b\cdot d(A,1)}{(n_a + n_b) d(B,1) + \sum_{i \in \V(A)} d(i, B)}\\
         &= 1 + \frac{2n_b\cdot d(A,1)}{(n_a + n_b) d(B,1) + n_a \cdot \mu}.\\
    \end{align*}
    Here recall that by our mechanism, since $A$ won the election, we have $n_b \cdot d(A,1) \leq n_a \cdot d(B,1)$, which means that $d(B,1) \geq \frac{n_b}{n_a} d(A,1)$. In addition, recall that we also have $\mu \geq d(A,1)$, so
    \begin{align*}
         \frac{\cost(M_A)}{\cost(M_B)} 
         &\leq 1 + \frac{2n_b\cdot d(A,1)}{(n_a + n_b) \frac{n_b}{n_a} d(A,1) + n_a\cdot d(A,1)}\\
         &\leq 1 + \frac{2n_b}{(n_a + n_b) \frac{n_b}{n_a} + n_a}.
     \end{align*}
     Now, we consider case (ii), $\mu \leq d(A,1)$. We have 
     \begin{align*}
         \frac{\cost(M_A)}{\cost(M_B)} &\leq 1 + \min \left\{\frac{2n_b\cdot d(A,1)}{\cost(M_B)}, \frac{2n_b\cdot \mu}{\cost(M_B)}\right\}\\
         &= 1 + \frac{2n_b\cdot \mu}{\cost(M_B)}\\
         &= 1 + \frac{2n_b\cdot \mu}{(n_a + n_b) d(B,1) + \sum_{i \in \V(A)} d(i, B) + \sum_{j \in \V(B)} d(j, B)}\\
         &\leq 1 + \frac{2n_b\cdot \mu}{(n_a + n_b) d(B,1) + \sum_{i \in \V(A)} d(i, B)}\\
         &= 1 + \frac{2n_b\cdot \mu}{(n_a + n_b) d(B,1) + n_a \cdot \mu}.\\
    \end{align*}

Similar to case (i), recall that by our mechanism, since $A$ won the election, we have $n_b \cdot d(A,1) \leq n_a \cdot d(B,1)$, which means that $d(B,1) \geq \frac{n_b}{n_a} d(A,1)$. Besides, since we have $d(A,1) \geq \mu$, we can see that $d(B,1) \geq \frac{n_b}{n_a} d(A,1) \geq \frac{n_b}{n_a} \mu$. Therefore, we have
    \begin{align*}
         \frac{\cost(M_A)}{\cost(M_B)} 
         &\leq 1 + \frac{2n_b\cdot \mu}{(n_a + n_b) \frac{n_b}{n_a} \cdot \mu + n_a \cdot \mu}\\
         &= 1 + \frac{2n_b}{(n_a + n_b) \frac{n_b}{n_a} + n_a}.
     \end{align*}
     Since both cases give the same result, we can conclude that 
     $$\frac{\cost(M_A)}{\cost(M_B)} 
         \leq 1 + \frac{2n_b}{(n_a + n_b) \frac{n_b}{n_a} + n_a}.$$
    However, since $n_a, n_b > 0$, by Lemma \ref{lem:53-function}, we have that 
    $$\frac{\cost(M_A)}{\cost(M_B)} \leq 1 + \frac{2n_b}{(n_a + n_b) \frac{n_b}{n_a} + n_a} \leq \frac{5}{3}$$
     Thus our mechanism gives a distortion that is at most $\frac{5}{3}$ as desired.
\end{proof}

We will then show that this bound is tight in the next theorem. 

\begin{thm}
        There is no deterministic mechanism for the \lineup{m}{1} problem knowing candidate position distance pairs and with voter preferences that achieves a better distortion better than $\frac{5}{3}$, even on a line and with only $2$ candidates.  \label{thm:lower-1-2-distances}
\end{thm}

\begin{figure}[h]
    \centering
    \begin{minipage}{.5\textwidth}
      \centering
    \includegraphics[]{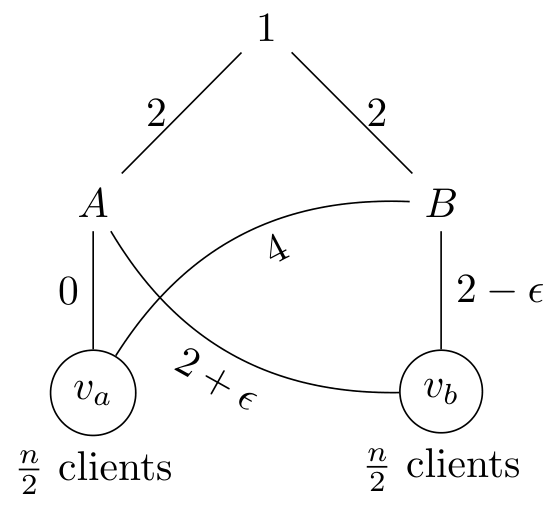}
    \subcaption{Metric space $d_1$.}
    \end{minipage}%
    \begin{minipage}{.5\textwidth}
      \centering
    \includegraphics[]{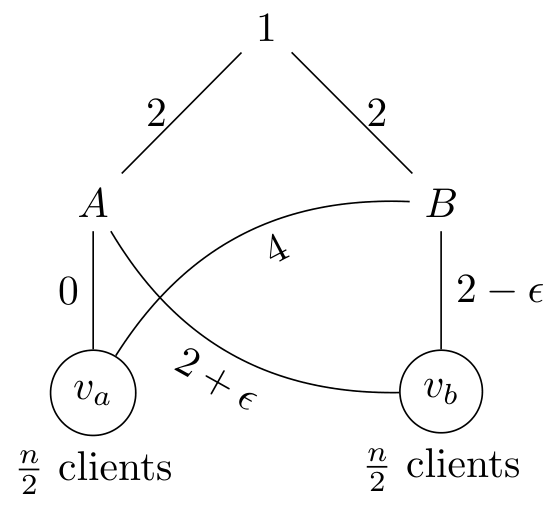}
    \subcaption{Metric space $d_2$.}
    \end{minipage}
    \caption{Two consistent metric spaces with distortion $\frac{5}{3}$.}
    \label{fig:thm52}
\end{figure}

\begin{proof}
    Consider the two instances in Figure \ref{fig:thm52}. Both instances have the same information set so if we pick $A$ then $d_2$ will have a distortion of $\lim_{\epsilon \to 0}\frac{\frac{n}{2}\left(2-\epsilon + 2 + 4 + 2\right)}{\frac{n}{2}\left(2+\epsilon + 2+2\right)} = \frac{5}{3}$ and if we pick $B$ then metric $d_1$ we get a distortion of $\lim_{\epsilon \to 0}\frac{\frac{n}{2}\left(2-\epsilon + 2 + 4 + 2\right)}{\frac{n}{2}\left(2+\epsilon + 2+2\right)} = \frac{5}{3}$.
\end{proof}

Theorem \ref{thm:lower-1-2-distances} means that the algorithm given in Theorem \ref{thm:mech-1-k-dist} achieves the best possible distortion bounds. We will now generalize the method in Theorem \ref{thm:mech-1-k-dist} to more than 2 candidates.

Assume there are $m$ candidates. We start by constructing a graph $G=(V,A)$ similar to a weighted majority graph. Instead we orient arcs such that instead of arcs going from winners according to a plurality election, they are oriented based on the winner given by Theorem \ref{thm:mech-1-k-dist}. So if $(x,y) \in A$ then $x$ was selected by the mechanism used in Theorem \ref{thm:mech-1-k-dist}. First we introduce the following lemma regarding the graph previously described. We then select the client with the largest out degree. 

\begin{lem}
    Let $O$ be the optimal candidate and $A$ be the candidate with largest out-degree. Then if $O \neq A$, there exists a path of at most length 2 from $A$ to $O$ in the graph defined above.
    \label{lem:path}
\end{lem}

\begin{proof}
    First we will show $d^+(A) \geq \frac{n-1}{2}$. Assume $d^+(A) < \frac{n-1}{2}$. If all vertices $v$ have $d^+(v) < \frac{n-1}{2}$ then the number of edges is strictly less than $n\frac{n-1}{2} = \frac{1}{2}n(n-1) = |A|$. So  $d^+(A) \geq \frac{n-1}{2}$.

    If there exists an arc from $A$ to $O$ then we are done since the arc forms a path of length 1. Therefore, suppose there is an arc from $O$ to $A$. Now, suppose there does not exist a path of length at most 2 from $A$ to $O$. Let $c$ be a vertex such that there is an arc from $A$ to $c$, then there cannot be an arc from $c$ to $O$ by our assumption and there exists at least $\frac{n-1}{2}$ such $c$. This means that there must an arc from $O$ to each of such $c$'s. Thus, combined with the arc $(O,A)$, we have $d^+(O) \geq \frac{n-1}{2}+1$. This means that $d^+(O) > d^+(A)$. But this is a contradiction since $A$ was chosen as the vertex with largest out-degree.
    
\end{proof}

We can now use the previous lemma to get a mechanism that gives a distortion of $\frac{25}{9}$.

\begin{thm}
    There exists a mechanism that achieves a distortion of $\frac{25}{9}$ for \lineup{m}{1}.
\end{thm}

\begin{proof}
    Suppose $A$ is the candidate that is chosen by the previous rule. Then by Lemma \ref{lem:path}, there are 3 cases. If $A$ is optimal we have a distortion of $1$. If there is a path of length 1 from $A$ to $O$ then we have a distortion of $\frac{5}{3}$. If there is a path of length two from $A$ to $O$ going through candidate $c$ then $\cost(A) \leq \frac{5}{3} \cost(c) \leq \frac{5}{3}\cdot \frac{5}{3} \cost(O) = \frac{25}{9} \cost(O)$.
\end{proof}

Recall from \citet{anshelevich2021ordinal} that in the standard election problem if we are given exact information about candidate locations, we cannot improve the distortion bounds beyond 3, even for 2 candidates, which is the same distortion achievable when given only voter preferences. In other words, for standard elections having candidate locations does not give any advantage in terms of worst-case distortion. As we have seen, however, for line-up elections this is no longer true: having candidate and position locations allows us to improve the distortion, at least for the case when $|\Pos|=1$. 

\section{Conclusion and Future Work}
   We defined metric line-up elections, and showed both upper bounds and lower bounds on distortion in these elections, for multiple types of information available to us. Our findings seem to indicate that only a small amount of targeted information can allow us to form good outcomes for such elections. Our distortion bounds outperform the ones for bipartite matching problems. They are also often better or at least comparable to distortion bounds in standard single-winner elections, despite the fact that in line-up elections we must choose multiple winners and then match them to positions. 

There still exist unanswered questions, however. Not all of our distortion bounds are tight: do better mechanisms exist? What about considering other cost functions in addition to social cost, for example the maximum cost? We also could look at settings where the cost for a voter is a subadditive combination of the distances (instead of the sum). Finally, what about randomized mechanisms? Can they provide distortion bounds much better than our deterministic mechanisms, unlike in the case for standard elections?
\clearpage
\bibliography{refs}
\appendix
\section{Further motivating examples} \label{sec:appendix}
We will now provide some additional examples to further motivate our model. 

{\em Example 1:} Consider the metric space of physical locations. We have a set of clients or customers $\V$, a set of possible distribution centers $\C$, and different types of products $\Pos$. The goal is to choose a set of distribution centers, and assign each product to exactly one distribution center. All products must be delivered through a distribution center to the clients. For example, one distribution center may be a location for goods like eggs and milk, which require refrigeration, being delivered to its grocery clients. Another distribution center may consist of goods that require dry bulk storage like wheat, another a location for fish, a third may be a location for equipment or tools, etc. The goal is to choose center locations so that the clients (customers) can easily and quickly obtain goods from all of them (i.e., minimize the distances $d(i,c)$ for $i\in \V$ and $c\in \C$). But we also want to choose a center location which makes it easy and convenient for them to obtain their products from producers: a center distributing dairy should be close to dairy farms, a center dealing in fish should be close to fisheries, and one dealing in manufactured goods should be close to the factories making those goods. This will affect the total amount of time and convenience for a customer in $\V$ to obtain the product. If each product $p\in \Pos$ has exactly one location where it is being produced, this becomes exactly the model in our paper: the cost for assigning product $p \in \Pos$ to distribution center $c \in \C$ is  $\sum_{i} \left[d(i, c) + d(c,p)\right]$ which is the same as the cost in a line-up election. Note, however, that even if each product is produced in several locations (i.e., $p\in \Pos$ is actually a set of locations), then our model still applies. Define $d(c,p) = \frac{1}{\abs{p}} \sum_{p' \in p} d(c,p')$ as the average distance. It is easy to check that for all $i\in V$, $c\in \C$, $p\in \Pos$, the distances $d(i,c)$ and $d(c,p)$ still form a metric space when defined as above (i.e., the distances $d(i,c)$ are the original distances, while $d(c,p)$ are the average distances from $c$ to all points in $p$). We can even assign multiple types of products to the same distribution center by duplicating each distribution center based on its capacity. All of this is a special case of line-up elections, and thus our results hold for this setting.

{\em Example 2:} Consider a large conference, in which we have attendees or more generally members of the research community $\V$. We also have papers under review or more generally research topics $\Pos$. A lot of previous work has considered the metric space of research interests and expertise: this space has a very high dimension, but every researcher, paper, or research topic can be loosely considered a point in this metric space. Now consider choosing senior program committee members, or area chairs, from some set of volunteers $\C$. We would like to choose an area chair for topic $p\in \Pos$ who is qualified, i.e., they should be located close to $p$ in this space.\footnote{Note that, just as in the previous example, $p$ can be a single point in the space, or a set of points representing different papers on this topic. In the latter case by defining the distance from $c\in \C$ to the set $p$ as the average distance to all the points in $p$, we still have a metric space and all our line-up election results hold.} For example, a certain participant may prefer that an area chair in charge of Social Choice at AAAI has expertise in the general topic of social choice, i.e., small average distance to the social choice submissions. On the other hand, someone may (perhaps selfishly) prefer that they are close to his or her personal research interests, since this makes it more likely for this person's interests to be represented and for this person's papers to be accepted. Thus if $c$ is chosen as the area chair for topic $p$, we would be interested in minimizing both $d(c,p)$ (their overall expertise), and $d(i,c)$ (how well they represent my sub-area specifically). This is exactly a setting of line-up elections.

\end{document}